\documentclass[letterpaper, 10 pt, conference]{ieeeconf}
\IEEEoverridecommandlockouts
\usepackage{amsmath}
\usepackage{amssymb}
\usepackage{mathptmx}
\usepackage{graphicx}
\usepackage{xfrac}
\usepackage{comment}
\usepackage{citesort}
\usepackage{color}
\usepackage{subfig}
\usepackage[hidelinks]{hyperref}

\newtheorem{definition}{Definition}

\newtheorem{theorem}{Theorem}

\newtheorem{remark}{Remark}

\title{\LARGE \bf Collaborative Visual Area Coverage using Unmanned Aerial Vehicles}
\author{Sotiris Papatheodorou, Anthony Tzes and Yiannis Stergiopoulos
	\thanks{This work has received funding from the European Union's Horizon 2020 Research and Innovation Programme under the Grant Agreement No.644128, AEROWORKS.}
\thanks{The authors are with the Electrical \& Computer Engineering Department, University of Patras, Rio, Achaia 26500, Greece. Corresponding author's email: \small\tt{tzes@ece.upatras.gr}}
}
\begin{document}
	\maketitle
	\thispagestyle{empty}
	\pagestyle{empty}
	%%%%%%%%%%%%%%%%%%%%%%%%%%%%%%%%%%%%%%%%%%%%%%%%%%%%%%%%%%%%%%%%%%%%%%%%%%%%%%%%
	\begin{abstract}
	This article addresses the visual area coverage problem using a team of Unmanned Aerial Vehicles (UAVs). The UAVs are assumed to be equipped with a downward facing camera covering all points of interest within a circle on the ground. The diameter of this circular conic-section increases as the UAV flies at a larger height, yet the quality of the observed area is inverse proportional to the UAV's height. The objective is to provide a distributed control algorithm that maximizes a combined coverage-quality criterion by adjusting the UAV's altitude. Simulation studies are offered to highlight the effectiveness of the suggested scheme.
	\end{abstract}
	

%%%%%%%%%%%%%%%%%%%%%%%%%%%%%%%%%%%%%%%%%%%%%%%%%%%%%%%%%%%%%%%%%%%%%%%%%%%%%%%%
\section{Introduction}
	Area coverage over a planar region by ground agents has been studied extensively when the sensing patterns of the robots are circular \cite{Cortes_ESAIMCOCV05,Pimenta_CDC08}. Most of these techniques are based on a Voronoi or similar partitioning~\cite{Stergiopoulos_IETCTA10,Arslan_ICRA2016} of the region of interest. However there have been studies concerning arbitrary sensing patterns~\cite{Stergiopoulos_Automatica13,Kantaros_Automatica15} where the partitioning is not Voronoi based~\cite{Stergiopoulos_ICRA14}. Both convex and non-convex domains have been examined~\cite{Stergiopoulos_IEEETAC15}.
	
	Many algorithms have been developed for the mapping by UAVs~\cite{Renzaglia_IJRR12,Breitenmoser_IROS10,Thanou_ISCCSP14} relying mostly in Voronoi-based tessellations. Extensive work has also been done in area monitoring by UAVs equipped with cameras~\cite{Schwager_ICRA2009,Schwager_IEEE2011}. In these pioneering research efforts, there is no maximum allowable height that can be reached by the UAVs and for the case where there is overlapping of the covered areas, this is considered an advantage compared to the same area viewed by a single camera.
	
	In this paper the persistent coverage problem of a planar region by a network of UAVs is considered. The UAVs are assumed to have downwards facing visual sensors with a circular sensing footprint, while no assumptions are made about the type of the sensor. The covered area as well as the coverage quality of that area depend on the altitude of each UAV. A partitioning scheme of the sensed region, similar to~\cite{Stergiopoulos_ICRA14}, is employed and a gradient based control law is developed. This control law leads the network to a locally optimal configuration with respect to a combined coverage-quality criterion, while also guaranteeing that the UAVs remain within a desired range of altitudes. This article, compared to~\cite{Schwager_IEEE2011}  assumes the flight of the UAVs within a given regime and attempts to provide visual information of an area using a single UAV-camera.
	
	The problem statement, along with the definition of the coverage-quality criterion and the sensed space partitioning scheme are presented in Section \ref{section:problem_statement}. In Section \ref{section:distributed_control_law}, the distributed control law is derived and its main properties are explained. In Section \ref{section:fi_examples}, the partitioning scheme and the control law are examined for two particular coverage quality functions. Finally, in Section \ref{section:simulation_studies},  simulations are provided to show the efficiency of the proposed technique.

%%%%%%%%%%%%%%%%%%%%%%%%%%%%%%%%%%%%%%%%%%%%%%%%%%%%%%%%%%%%%%%%%%%%%%%%%%%%%%%%
\section{Problem Statement}
\label{section:problem_statement}

	\subsection{Main assumptions - Preliminaries}
	Let $\Omega \subset \mathbb{R}^2$ be a compact convex region under surveillance. Assume a swarm of $n$ UAVs, each positioned at the spatial coordinates $X_i = \left[ x_i,y_i, z_i \right]^T, ~i \in I_n$, where $I_n = \left\{ 1, \dots ,n \right\}$. We also define the vector $q_i = [ x_i ~ y_i ]^T, ~q_i \in \Omega$ to note the projection of the center of each UAV on the ground. The minimum and maximum altitudes each UAV can fly to are $z_i^{\min}$ and $z_i^{\max}$ respectively, thus $z_i \in [z_i^{\min}, ~z_i^{\max}], i \in I_n$. 
	
	The simplified UAV's kinodynamic model is 
	\begin{eqnarray}
		\nonumber
		\dot{q}_i &=& u_{i,q},~~q_i \in \Omega, ~u_{i,q} \in \mathbb{R}^2,\\
		\dot{z}_i &=& u_{i,z},~~z_i \in [z_i^{\min} ~,~ z_i^{\max}], ~u_{i,z}\in \mathbb{R}.
		\label{kinematics}
	\end{eqnarray}
	where $\left[u_{i,q}, u_{i,z}\right]$ is the corresponding `thrust' control input for each robot~(node). The minimum altitude $z_i^{\min}$ is used to ensure the UAVs will fly above ground obstacles, whereas the maximum altitude $z_i^{\max}$ guarantees that they will not fly out of range of their base station. In the sequel, all UAVs are assumed to have common minimum $z^{\min}$ and maximum $z^{\max}$ altitudes.

	As far as the sensing performance of the UAVs ~(nodes) is concerned, all members are assumed to be equipped with identical downwards pointing sensors with conic sensing patterns. Thus the region of $\Omega$ sensed by each node is a disk defined as
	\begin{equation}
		C_{i}^{s}(X_i,a) = \left\{ q \in \Omega: \parallel q-q_i\parallel \leq z_i~ \tan a \right\},~i=1,\ldots,n,
		\label{sensing}
	\end{equation}
	where $a$ is half the angle of the sensing cone. As shown in Figure \ref{fig:problem_statement}, the higher the altitude of a UAV, the larger the area of $\Omega$ surveyed by its sensor. 
	
	The coverage quality of each node is in the general case a function $f_i(q)\colon \mathbb{R}^2 \rightarrow [0, 1]$ which is dependent on the node's position $X_i$ as well as its altitude constraints $z^{\min}, z^{\max}$ and the angle $a$ of its sensor. The higher the value of $f_i(q)$, the better the coverage quality of point $q \in \mathbb{R}^2$, with $f_i = 1$ when $z_i = z^{\min}$ and $f_i = 0$ when $z_i = z^{\max}$. It is assumed that as the altitude of a node increases, the visual quality of its sensed area decreases. The properties of $f_i$ are examined in Section \ref{section:fi_examples} along with some example functions. 

	For each point $q \in \Omega$, an importance weight is assigned via the space density function $\phi \colon \Omega \rightarrow \mathbb{R}^+$, encapsulating any a priori information regarding the region of interest. Thus the coverage-quality objective is
	\begin{equation}
	\mathcal{H} \stackrel{\triangle}{=} \int_{\Omega} \max_{i \in I_n} f_i(q) ~\phi(q) ~dq.
	\end{equation}
	In the sequel, we assume $\phi(q) = 1, ~\forall q \in \Omega$ but the expressions can be easily altered to take into account any a priori weight function.

	\begin{figure}[htb]
		\centering
		\includegraphics[width=0.3\textwidth]{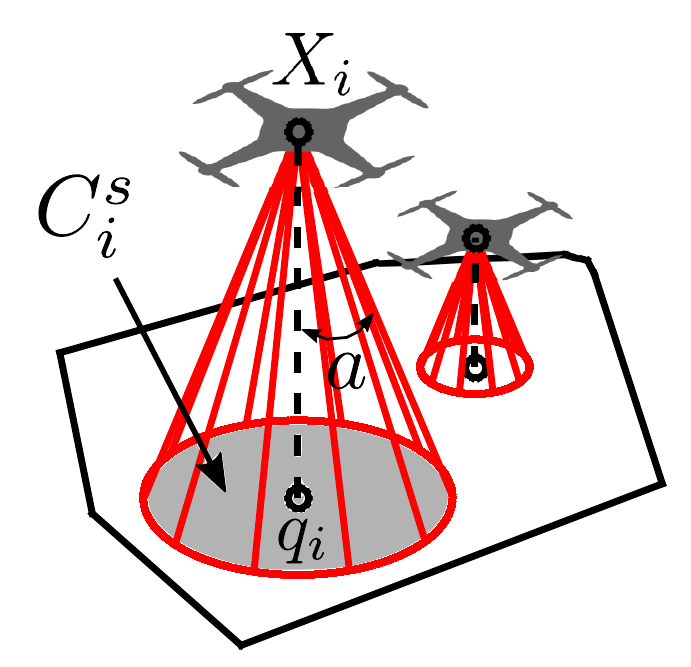}
		\caption{UAV-visual area coverage concept}
		\label{fig:problem_statement}
	\end{figure}

	\subsection{Sensed space partitioning}
	\label{section:partitioning}
	The assignment of responsibility regions to the nodes is done in a manner similar to \cite{Stergiopoulos_ICRA14}. Only the subset of $\Omega$ sensed by the nodes is partitioned. Each node is assigned a cell
	\begin{equation}
	W_i \stackrel{\triangle}{=} \left\{q \in \Omega \colon f_i(q) \geq f_j(q), ~j \neq i \right\}
	\label{partitioning}
	\end{equation}
	with the equality holding true only at the boundary $\partial W_i$, so that the cells $W_i$ comprise a complete tessellation of the sensed region.
	
	\begin{definition}
	We define the neighbors $N_i$ of node $i$ as
	\begin{equation*}
	N_i \stackrel{\triangle}{=} \left\{ j \neq i \colon C_j^s \cap C_i^s \neq \emptyset \right\}.
	\end{equation*}
	The neighbors of node $i$ are those nodes that sense at least a part of the region that node $i$ senses. It is clear that only the nodes in $N_i$ need to be considered when creating $W_i$.
	\end{definition}
	\begin{remark}
	The aforementioned partitioning is a complete tessellation of the sensed region $\bigcup_{i \in I_n} C_i^s$, although it is not a complete tessellation of $\Omega$. Let us denote the neutral region not assigned by the partitioning scheme as $\mathcal{O} = \Omega \setminus \bigcup_{i \in I_n} W_i$.
	\end{remark}
	\begin{remark}
	The resulting cells $W_i$ are compact but not necessarily convex. It is also possible that a cell $W_i$ consists of multiple disjoint regions, such as the cell of node $1$ shown in yellow in Figure \ref{fig:disjoint_empty_cells}. It is also possible that the cell of a node is empty, such as node $8$ in Figure \ref{fig:disjoint_empty_cells}. Its sensing circle $\partial C_8^s$ is shown in a dashed black line.
	\end{remark}
	
	\begin{figure}[htb]
		\centering
		\includegraphics[width=0.3\textwidth]{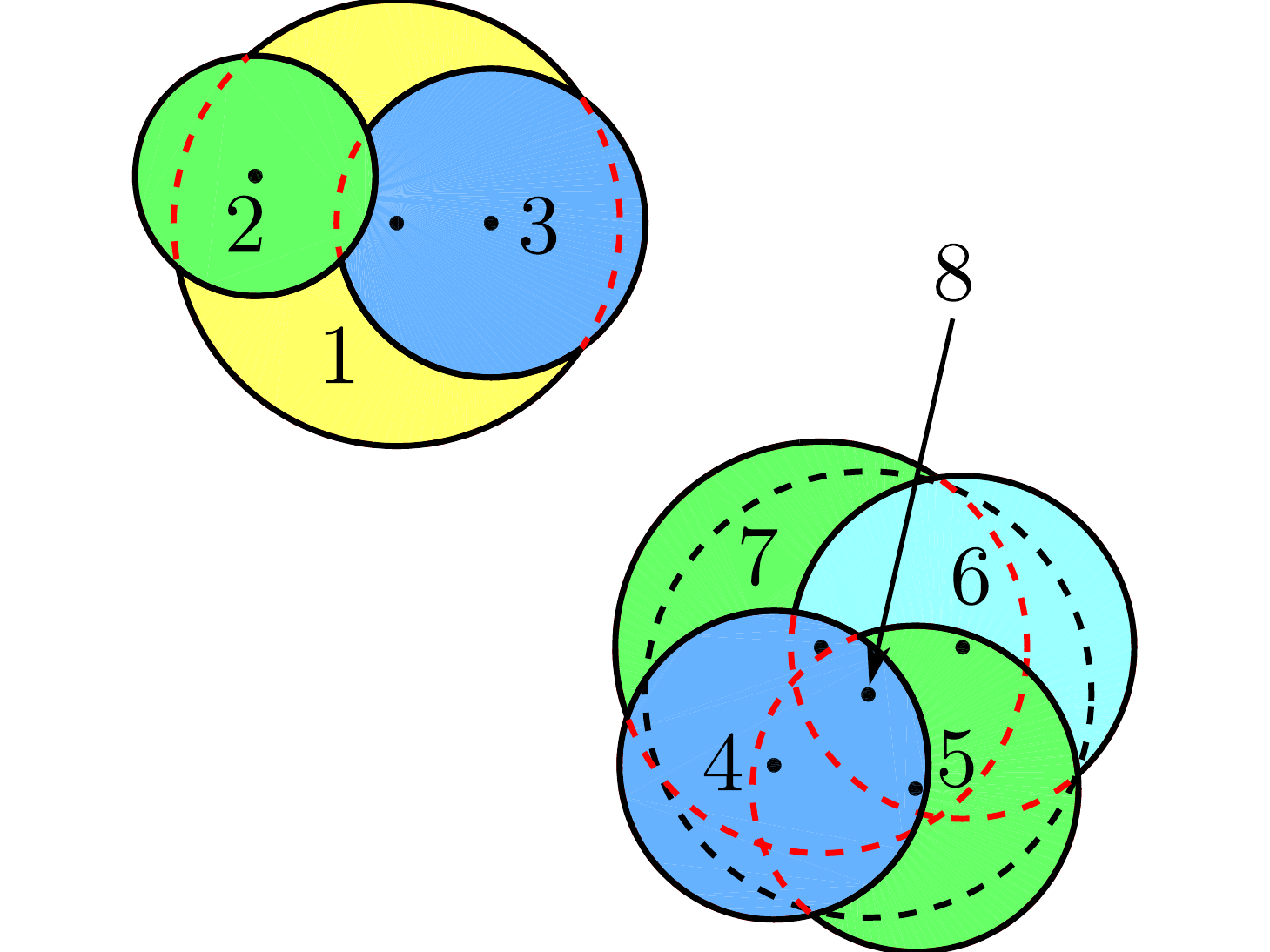}
		\caption{Examples of an empty cell (node 8) and one consisting of two disjoint regions (node 1).}
		\label{fig:disjoint_empty_cells}
	\end{figure}
		
	By utilizing this partitioning scheme, the network's coverage performance can be written as
	\begin{equation}
	\mathcal{H} = \sum_{i \in I_n} \int_{W_i} f_i(q) ~\phi(q) ~dq.
	\label{criterion}
	\end{equation}

%%%%%%%%%%%%%%%%%%%%%%%%%%%%%%
\section{Spatially Distributed Coordination Algorithm}
%%%%%%%%%%%%%%%%%%%%%%%%%%%%%%
\label{section:distributed_control_law}

	Based on the nodes kinodynamics (\ref{kinematics}), their sensing performance (\ref{sensing}) and the coverage criterion (\ref{criterion}), a gradient based control law is designed. The control law utilizes the partitioning (\ref{partitioning}) and result in monotonous increase of the covered area.
	\begin{theorem}
	In a UAV visual network consisting of nodes with sensing performance as in (\ref{sensing}), governed by the kinodynamics in (\ref{kinematics})
	 and the space partitioning described in Section \ref{section:partitioning}, the control law
	\begin{eqnarray}
		\nonumber
		u_{i,q} &=& \alpha_{i,q} \left[ ~\int\limits_{\partial W_i \cap \partial \mathcal{O}} n_i ~f_i(q) ~dq ~+ \int\limits_{W_i} \frac{\partial f_i(q)}{\partial q_i} ~dq ~+ \right.\\
		& & \left. \sum\limits_{j \neq i} ~\int\limits_{\partial W_i \cap \partial W_j} \upsilon_i^i ~n_i ~(f_i(q) - f_j(q)) ~dq \right]\\
		\nonumber
		u_{i,z} &=& \alpha_{i,z} \left[ ~\int\limits_{\partial W_i \cap \partial \mathcal{O}} \tan(a) ~f_i(q) ~dq ~+ \int\limits_{W_i} \frac{\partial f_i(q)}{\partial z_i} ~dq ~+ \right.\\
		& & \left. \sum\limits_{j \neq i} ~\int\limits_{\partial W_i \cap \partial W_j} \nu_i^i \cdot n_i ~(f_i(q) - f_j(q)) ~dq \right]
		\label{control_law}
	\end{eqnarray}
	where $\alpha_{i,q}, \alpha_{i,z}$ are positive constants and $n_i$ the outward pointing normal vector of $W_i$, maximizes the performance criterion (\ref{criterion}) monotonically along the nodes' trajectories, leading in a locally optimal configuration.
	\end{theorem}
	
	\begin{proof}
	Initially we evaluate the time derivative of the optimization criterion $\mathcal{H}$
	\begin{equation*}
	\frac{d\mathcal{H}}{dt} = \sum_{i \in I_n} \left[ \frac{\partial\mathcal{H}}{\partial q_i} \dot{q}_i ~+ \frac{\partial\mathcal{H}}{\partial z_i} \dot{z}_i\right].
	\end{equation*}.
	
	The usage of a gradient based control law in the form
	\begin{equation*}
	u_{i,q} = \alpha_{i,q} \frac{\partial\mathcal{H}}{\partial q_i}, ~~ u_{i,z} = \alpha_{i,z} \frac{\partial\mathcal{H}}{\partial z_i}
	\end{equation*}
	will result in a monotonous increase of $\mathcal{H}$.
	
	By using the Leibniz integral rule \cite{Flanders_AMM73} we obtain

	\begin{eqnarray}
	\tiny
	\nonumber
	\frac{\partial\mathcal{H}}{\partial q_i} &=& \sum\limits_{i \in I_n}
	\left[
	~\int\limits_{\partial W_i} \upsilon_i^i ~n_i ~f_i(q) ~dq ~+ \int\limits_{W_i} \frac{\partial f_i(q)}{\partial q_i} ~dq
	\right]\\
	\nonumber
	&=& ~\int\limits_{\partial W_i} \upsilon_i^i ~n_i ~f_i(q) ~dq ~+ \int\limits_{W_i} \frac{\partial f_i(q)}{\partial q_i} ~dq +\\
	\nonumber
	& & \sum\limits_{j \neq i} \left[
	~\int\limits_{\partial W_j} \upsilon_j^i ~n_j ~f_j(q) ~dq ~+ \int\limits_{W_j} \frac{\partial f_j(q)}{\partial q_i} ~dq \right]
	\normalsize
	\end{eqnarray}
	where $\upsilon_j^i$ stands for the Jacobian matrix with respect to $q_i$ of the points $q\in \partial W_j$, 
	\begin{equation}
	\upsilon_j^i\left(q\right) \stackrel{\triangle}{=} \frac{\partial q}{\partial q_i},~~q\in \partial W_j,~i,j\in I_n.
	\end{equation}
	Since $\frac{\partial f_j(q)}{\partial q_i} = 0$ we obtain
	\begin{eqnarray}
	\tiny
	\nonumber
	\frac{\partial\mathcal{H}}{\partial q_i}
	&=& ~\int\limits_{\partial W_i} \upsilon_i^i ~n_i ~f_i(q) ~dq ~+ \int\limits_{W_i} \frac{\partial f_i(q)}{\partial q_i} ~dq +\\
	\nonumber
	& & \sum\limits_{j \neq i}
	~\int\limits_{\partial W_j} \upsilon_j^i ~n_j ~f_j(q) ~dq
	\normalsize
	\end{eqnarray}
	whose three terms indicate how a movement of node $i$ affects the boundary of its cell, its whole cell and the boundaries of the cells of other nodes. It is clear that only the cells $W_j$ which have a common boundary with $W_i$ will be affected and only at that common boundary.
	
	The boundary $\partial W_i$ can be decomposed in disjoint sets as
	\begin{equation}
	\hspace{-0.05cm}\partial W_i=
	\left\{\partial W_i \cap \partial \Omega \right\}
	\cup 
	\left\{\partial W_i \cap \partial \mathcal{O} \right\}
	\cup
	\{\bigcup_{j \neq i} \left( \partial W_i \cap \partial W_j \right)\}.
	\label{boundary_decomposition}
	\end{equation}
	These sets represent the parts of $\partial W_i$ that lie on the boundary of $\Omega$, the boundary of the node's sensing region and the parts that are common between the boundary of the cell of node $i$ and those of other nodes. This decomposition can be seen in Figure~\ref{fig:boundary_decomposition} with the sets $\partial W_i \cap \partial \Omega$, $\partial W_i \cap \partial \mathcal{O}$ and $\partial W_i \cap \bigcup_{j \neq i} \partial W_j$ appearing in solid red, green and blue respectively.
	
	\begin{figure}[htb]
		\centering
		\includegraphics[width=0.4\textwidth]{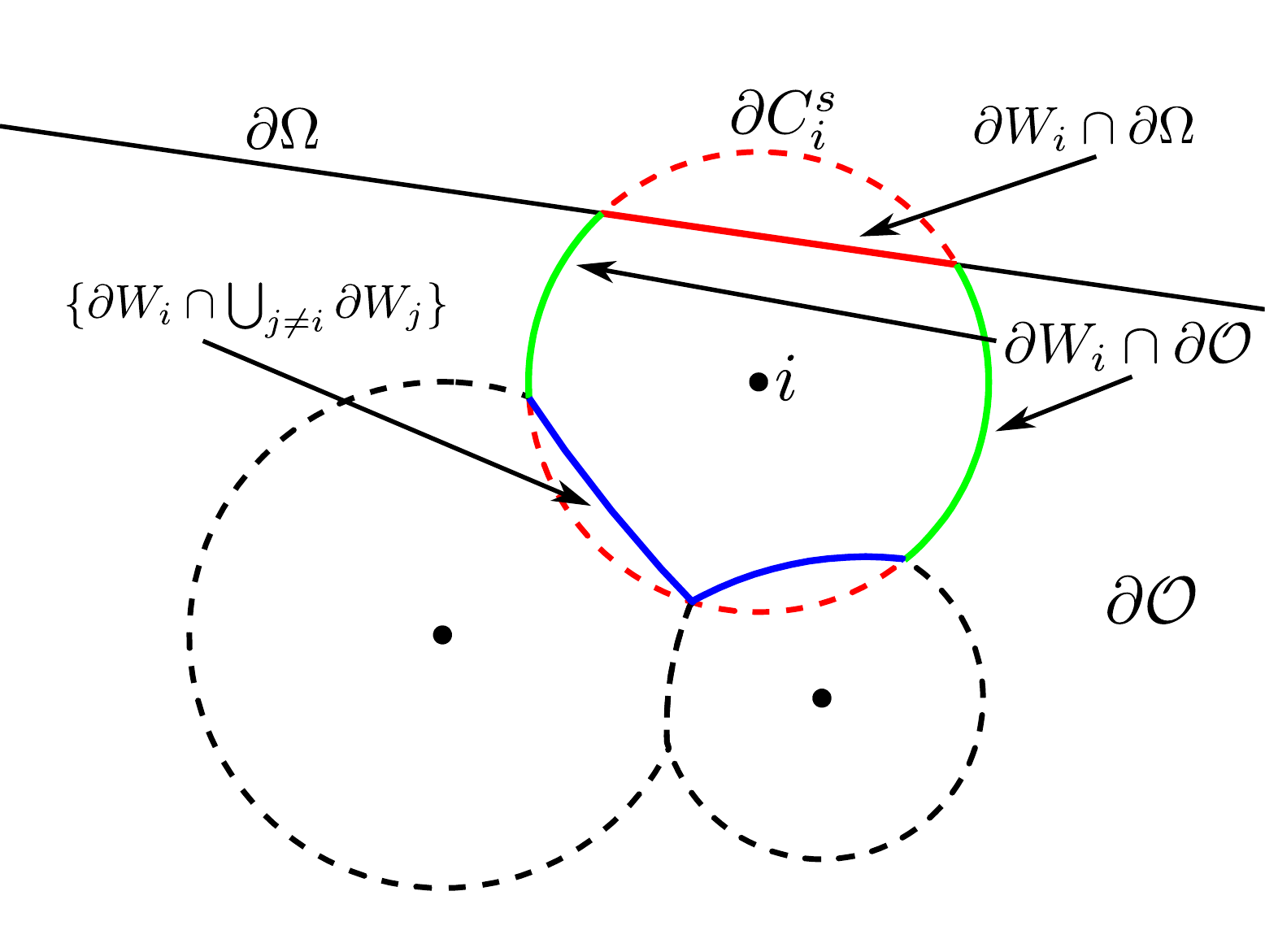}
		\caption{$\partial W_i$-decomposition into disjoint sets}
		\label{fig:boundary_decomposition}
	\end{figure}
	
	At $q \in \partial \Omega$ it holds that $\upsilon_i^i = \textbf{0}_{2 \times 2}$ since we assume the region of interest is static. Additionally, since only the common boundary $\partial W_j \cap \partial W_i$ of node $i$ with any other node $j$ is affected by the movement of node $i$, $\frac{\partial\mathcal{H}}{\partial q_i}$ can be simplified as
	\begin{eqnarray}
	\tiny
	\nonumber
	\frac{\partial\mathcal{H}}{\partial q_i}
	&=& ~\int\limits_{\partial W_i \cap \partial \mathcal{O}} \upsilon_i^i ~n_i ~f_i(q)  ~dq ~+ \int\limits_{W_i} \frac{\partial f_i(q)}{\partial q_i} ~dq ~+\\
	\nonumber
	& & \sum\limits_{j \neq i}
	~\int\limits_{\partial W_i \cap \partial W_j} \upsilon_i^i ~n_i ~f_i(q) ~dq ~+\\
	\nonumber
	& & \sum\limits_{j \neq i}
	~\int\limits_{\partial W_j \cap \partial W_i} \upsilon_j^i ~n_j ~f_j(q) ~dq.
	\normalsize
	\end{eqnarray}
	
	The evaluation of $\upsilon_i^i$ on $\partial W_i \cap \partial \mathcal{O}$ can be found in the Appendix, whereas its evaluation on $\partial W_j \cap \partial W_i$ depends on the choice of $f_i(q)$ and is not examined in this section. Because the boundary $\partial W_i \cap \partial W_j$ is common among nodes $i$ and $j$, it holds true that $\upsilon_j^i = \upsilon_i^i$ when evaluated over it and that  $n_j = -n_i$. Finally the sums and the integrals within then can be combined, producing the final form of the planar control law
	\begin{eqnarray}
	\tiny
	\nonumber
	\frac{\partial\mathcal{H}}{\partial q_i}
	&=& ~\int\limits_{\partial W_i \cap \partial \mathcal{O}} ~n_i ~f_i(q)  ~dq ~+ \int\limits_{W_i} \frac{\partial f_i(q)}{\partial q_i} ~dq ~+\\
	\nonumber
	& & \sum\limits_{j \neq i}
	~\int\limits_{\partial W_j \cap \partial W_i} \upsilon_i^i ~n_i ~\left(f_i(q) - f_j(q)\right) ~dq.
	\normalsize
	\end{eqnarray}
	
	Similarly, by using the same $\partial W_i$ decomposition and defining $\nu_j^i\left(q\right) \stackrel{\triangle}{=} \frac{\partial q}{\partial z_i},~~q\in \partial W_j,~i,j\in I_n$, the altitude control law is
	\begin{eqnarray}
	\tiny
	\nonumber
	\frac{\partial\mathcal{H}}{\partial z_i}
	&=& ~\int\limits_{\partial W_i \cap \partial \mathcal{O}} ~\tan(a) ~f_i(q)  ~dq ~+ \int\limits_{W_i} \frac{\partial f_i(q)}{\partial z_i} ~dq ~+\\
	\nonumber
	& & \sum\limits_{j \neq i}
	~\int\limits_{\partial W_j \cap \partial W_i} \nu_i^i \cdot n_i ~\left(f_i(q) - f_j(q)\right) ~dq
	\normalsize
	\end{eqnarray}
	where the evaluation of $\nu_i^i(q) \cdot n_i$ on $\partial W_i \cap \partial \mathcal{O}$ and $\partial W_j \cap \partial W_i$ can be found in the Appendix.
	
	\end{proof}
	
	\begin{remark}
	The cell $W_i$ of node $i$ is affected only by its neighbors $N_i$ thus resulting in a distributed control law. The finding of the neighbors $N_i$ depends on their coordinates $X_j,~j \in N_{i}$ and does not correspond to the classical 2D-Delaunay neighbors. The computation of the $N_i$ set demands node $i$ to be able to communicate with all nodes within a sphere centered around $X_i$ and radius $r_{i}^{c}$
	\small
	\begin{eqnarray}
	\nonumber
	r_{i}^{c}=\max\left\{2 z_i~\tan a, 
	~\left( z_i+z^{\min}\right)^2 \tan^2 a + \left( z_ - z^{\min} \right)^2,\right. \\
	\nonumber
	\left. 
	~\left( z_i+z^{\max}\right)^2 \tan^2 a + \left( z_ - z^{\max} \right)^2 \right\}.
	\end{eqnarray}
	\normalsize
	Figure~\ref{fig:neighbor_set} highlights the case where nodes $2$, $3$ and $4$ are at $z^{\min}$, $z_1$ and $z^{\min}$ respectively. These are the worst case scenario neighbors of node $1$ , the farthest of which dictates the communication range $r_1^c$. 

	\begin{figure}[htb]
		\centering
		\includegraphics[width=0.2\textwidth]{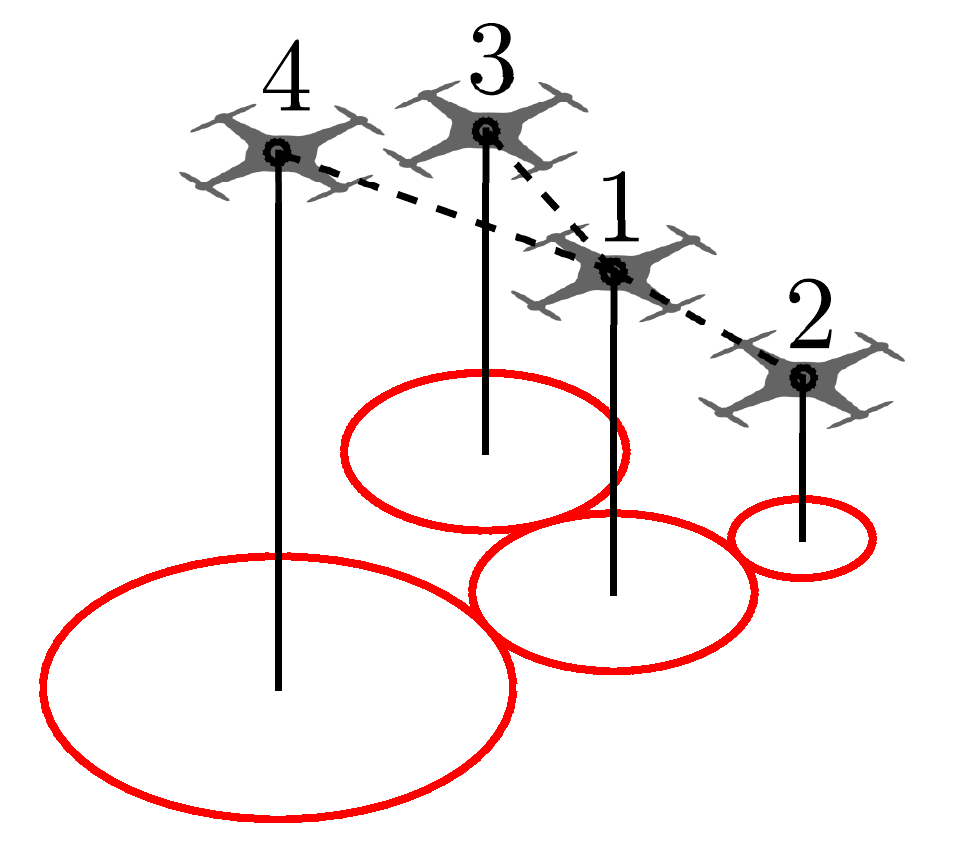}
		\caption{$N_i$ neighbor set}
		\label{fig:neighbor_set}
	\end{figure}
	\end{remark}
	
	\begin{remark}
	\label{remark:zmax}
	When $z_i = z^{\max}$, both the planar and altitude control laws are zero because $f_i(q_i) = 0$. This results in the UAV being unable to move any further in the future and additionally its contribution to the coverage-quality objective being zero. However this degenerate case is of little concern, as shown in Sections \ref{section:stable_alt} and \ref{section:degenerate_cases}.
	\end{remark}
	
	\subsection{Stable altitude}
	\label{section:stable_alt}
	The altitude control law $u_{i,z}$ moves each UAV towards an altitude in which $u_{i,z} = 0$ which corresponds to an equilibrium point for that particular node. We call this the stable altitude $z_i^{stb}$ and it is the solution with respect to $z_i$ of the equation:
	\begin{eqnarray}
	\nonumber
	&u_{i,z} = 0 \Rightarrow\\
	\nonumber
	&~\int\limits_{\partial W_i \cap \partial \mathcal{O}} \tan(a) ~f_i(q) ~dq ~+ \int\limits_{W_i} \frac{\partial f_i(q)}{\partial z_i} ~dq ~+\\
	\nonumber
	&\sum\limits_{j \neq i} ~\int\limits_{\partial W_i \cap \partial W_j} \nu_i^i \cdot n_i ~(f_i(q) - f_j(q)) ~dq = 0.
	\end{eqnarray}
	Both integrals over $\partial W_i$ are non-negative since their integrands are non-negative, whereas the integral over $W_i$ is negative since coverage quality decreases as altitude increases.
	
	The stable altitude is not common among nodes as it depends on one's neighbors $N_i$ and is not constant over time since the neighbors change over time.
	
	When the integrals over $\partial W_i$ are both zero, the resulting control law has a negative value. This will lead to a reduction of the node's altitude and in time the node will reach $z^{stb} = z^{\min}$. Once the node reaches $z^{\min}$ its altitude control law will be $0$ until the integral over $\partial W_i$ stops being zero. The planar control law however is unaffected, so the node's performance in the future is not affected. This situation may arise in a node with several neighboring nodes at lower altitude that result in $\partial C_i^s \cap \partial W_i = \emptyset$.
	
	When the integral over $W_i$ is zero, the resulting control law has a positive value. This will lead to an increase of the node's altitude and in time the node will reach $z^{stb} = z^{\max}$ and as shown in Remark \ref{remark:zmax} the node will be immobilized from this time onwards. However this situation will not arise in practice as explained in Section \ref{section:degenerate_cases}.
	
	When the integral over $W_i$ and at least one of the integrals over $\partial W_i$ are non-zero, then ${z^{stb} \in \left(z^{\min}, z^{\max}\right)}$.
	
	\subsection{Optimal altitude}
	It is useful to define an optimal altitude $z^{opt}$ as the altitude a node would reach if: 1) it had no neighbors $(N_i = \varnothing)$, and 2) no part of $W_i$ on $\partial \Omega,~(\partial \Omega \cap \partial C_i^s = \varnothing)$.
	This optimal altitude is the solution with respect to $z_i$ of the equation
	\begin{eqnarray}
	\nonumber
	\int\limits_{\partial C_i^s} \tan(a) ~f_i(q) ~dq ~+ \int\limits_{W_i} \frac{\partial f_i(q)}{\partial z_i} ~dq = 0.
	\end{eqnarray}
	Additionally, let us denote the sensing region of a node $i$ at $z^{opt}$ as $C_{i,opt}^s$ and $\mathcal{H}_{opt}$ the value of the criterion when all nodes are located at $z^{opt}$.
	
	If $\Omega = \mathbb{R}^2$ and because the planar control law $u_{i,q}$ results in the repulsion of the nodes, the network will reach a state in which no node will have neighbors and all nodes will be at $z^{opt}$. In that state, the coverage-quality criterion (\ref{criterion}) will have attained it's maximum possible value $\mathcal{H}_{opt}$ for that particular network configuration and coverage quality function $f_i$. This network configuration will be globally optimal. 
	
	When $\Omega$ is a convex compact subset of $\mathbb{R}^2$, it is possible for the network to reach a state where all the nodes are at $z^{opt}$ only if $n$ $C_{i,opt}^s$ disks can be packed inside $\Omega$. This state will be globally optimal. If that is not the case, the nodes will converge at some altitude other than $z^{opt}$ and in general different among nodes. It should be noted that although the nodes do not reach $z^{opt}$, the network configuration is locally optimal.

	\subsection{Degenerate cases}
	\label{section:degenerate_cases}
	It is possible due to the nodes' initial positions that the sensing disk of some node $i$ is completely contained within the sensing disk of another node $j$, i.e. $C_i^s \cap C_j^s = C_i^s$. In such a case, it is not guaranteed that the control law will result in separation of the nodes' sensing regions and thus it is possible that the nodes do not reach $z^{opt}$. Instead, node $j$ may converge to a higher altitude and node $i$ to a lower altitude than $z^{opt}$, while their projections on the ground $q_i$ and $q_j$ remain stationary. Because the region covered by node $i$ is also covered by node $j$, the network's performance is impacted negatively. Since this degenerate case may only arise at the network's initial state, care must be taken to avoid it during the robots' deployment. Such a degenerate case is shown in Figure \ref{fig:partitioning_comparison} where the sensing disk of node $4$ is completely contained within that of node $3$.
	
	Another case of interest is when some node $i$ is not assigned a region of responsibility, i.e. ${W_i  = \emptyset}$. This is due to the existence of other nodes at lower altitude that cover all of $C_i^s$ with better quality than node $i$. This is the case with node $8$ in Figure \ref{fig:disjoint_empty_cells}. This situation is resolved since the nodes at lower altitude will move away from node $i$ and once node $i$ has been assigned a region of responsibility it will also move. It should be noted that the coverage objective $\mathcal{H}$ remains continuous even when node $i$ changes from being assigned no region of responsibility to being assigned some region of responsibility.
	
	In order for a node to reach $z^{\max}$, as explained in Section \ref{section:stable_alt}, the integral over $W_i$ of its altitude control law must be zero, that is its cell must consist of just a closed curve without its interior. In order to have $W_i = \partial W_i$, a second node $j$ must be directly below node $i$ at an infinitesimal distance. However just as node $i$ starts moving upwards the integral over $W_i$ will stop being zero thus changing the stable altitude to some value $z^{stb} < z^{\max}$. In other words, in order for a node to reach $z^{\max}$, the configuration described must happen at an altitude infinitesimally smaller than $z^{\max}$. So in practice, if all nodes are deployed initially at an altitude smaller than $z^{\max}$, no node will reach $z^{\max}$ in the future.

	\begin{figure}[htb]
		\centering
		\includegraphics[width=0.23\textwidth]{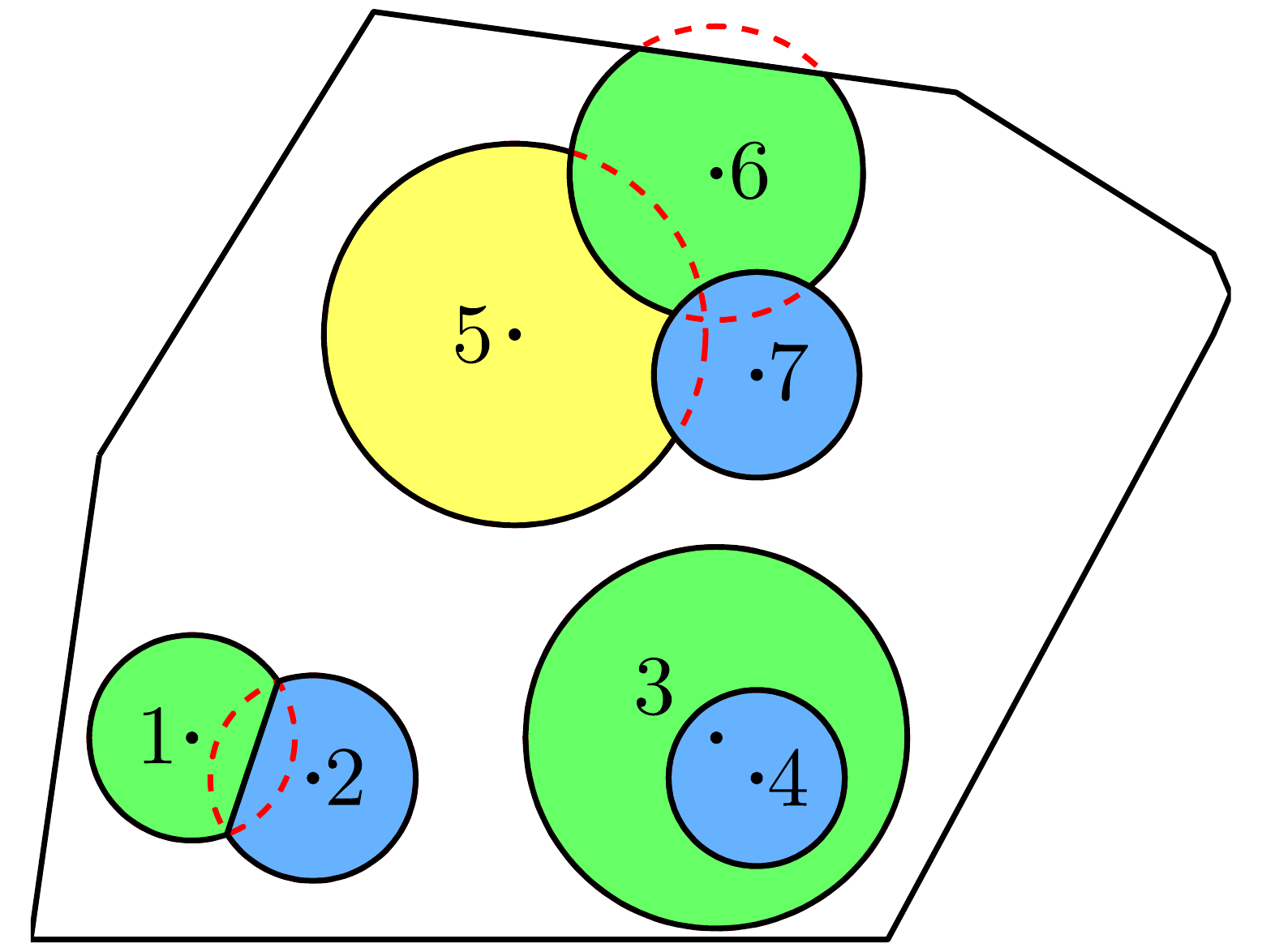}\hspace{0.01cm}
		\includegraphics[width=0.23\textwidth]{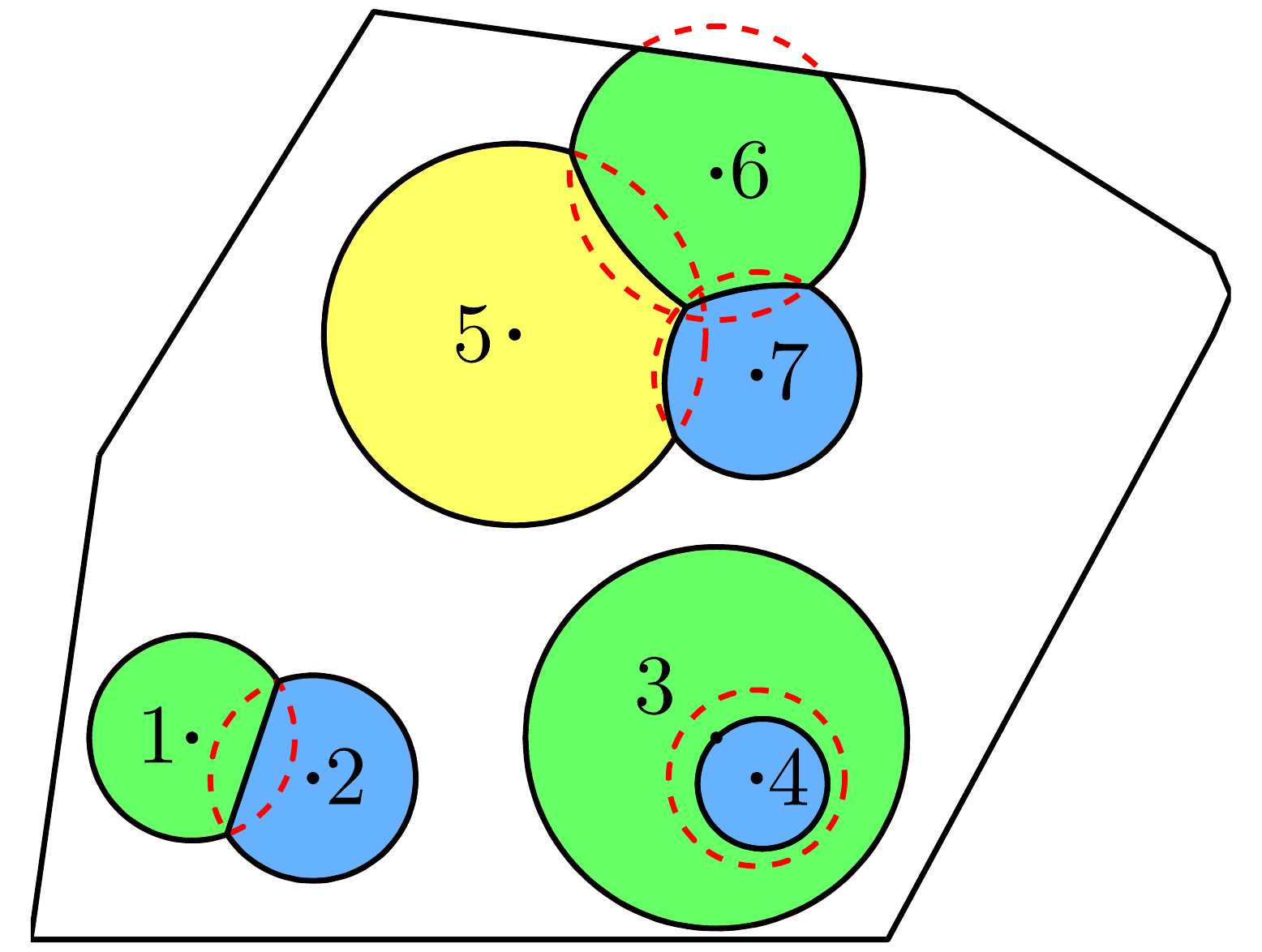}
		\caption{Sensed space partitioning using the coverage quality functions $f_i^u$ [Left] and $f_i^p$ [Right].}
		\label{fig:partitioning_comparison}
	\end{figure}

%%%%%%%%%%%%%%%%%%%%%%%%%%%%%%%%%%%%%%%%%
\section{Coverage quality functions}
\label{section:fi_examples}
%%%%%%%%%%%%%%%%%%%%%%%%%%%%%%%%%%%%%%%%%
	The choice of coverage quality function affects both the sensed space partitioning (\ref{partitioning}) and the control law (\ref{control_law}). The function $f_i,~i \in I_n$ is required to have the following properties
	\begin{enumerate}
	\item $f_i(q) = 0, ~\forall q \notin C_i^s, $
	\item $f_i(q) \geq 0, ~\forall q \in \partial C_i^s,$
	\item $f_i(q)$ is first order differentiable with respect to $q_i$ and $z_i$, or $\frac{\partial f_i(q)}{\partial q_i}$ and $\frac{\partial f_i(q)}{\partial z_i}$ exist within $C_i^s$,
	\item $f_i(q)$ is symmetric around the $z$-axis,
	\item $f_i(q)$ is a decreasing function of $z_i$,
	\item $f_i(q)$ is a non-increasing function of $\parallel q-q_i\parallel$,
	\item $f_i(q_i) = 1$ when $z_i = z^{\min}$ and $f_i(q_i) = 0$ when $z_i = z^{\max}$.
	\end{enumerate}
	The need for the first property is clear since the coverage quality with respect to node $i$ should be zero outside $C_i^s$. The second guarantees that the integrals of the control law over some part of $\partial C_i^s$ are not zero and the third is required so that the integrals over $W_i$ can be defined. The fourth property is due to the fact that the sensing patterns are rotation invariant. Two indicative coverage quality functions are examined in the sequel along with the properties of the resulting partitioning schemes and control laws. 
	
	\subsection{Uniform coverage quality}
	The simplest coverage quality function that can be chosen is the uniform one, i.e. the coverage quality is the same for all points in $C_i^s$. This function can effectively model downward facing cameras~\cite{DiFranco_JIRS16,Avellar_S2015} that provide uniform quality in the whole image. This function will be referred to as $f_i^u$ from now on and is defined as
	\begin{equation*}
	f_i^u(q) = \left \{
	\begin{aligned}
		&~\frac{\left( \left( z_i - z^{\min} \right)^2 - \left( z^{\max} - z^{\min} \right)^2 \right)^2}{\left( z^{\max} - z^{\min} \right)^4}, & ~q \in C_i^s\\
		&~0, & ~q \notin C_i^s
	\end{aligned}
	\right.\\
	\end{equation*}
	It should be noted that the above definition of a uniform coverage quality function is not unique, just the simplest of the possible ones.
	
	A projection of this function on the $y-z$ plane can be seen in the left portion of Figure \ref{fig:quality_functions} for a node at $X_i = \left[ 0,0, z^{\min} \right]^T$. It can be shown that $f_i^u(q) = 0$ when $\parallel q-q_i \parallel > z_i ~\tan{a}$, i.e. when $q$ is outside the sensing disk.
	
	\begin{figure}[htb]
		\centering
		\includegraphics[width=0.23\textwidth]{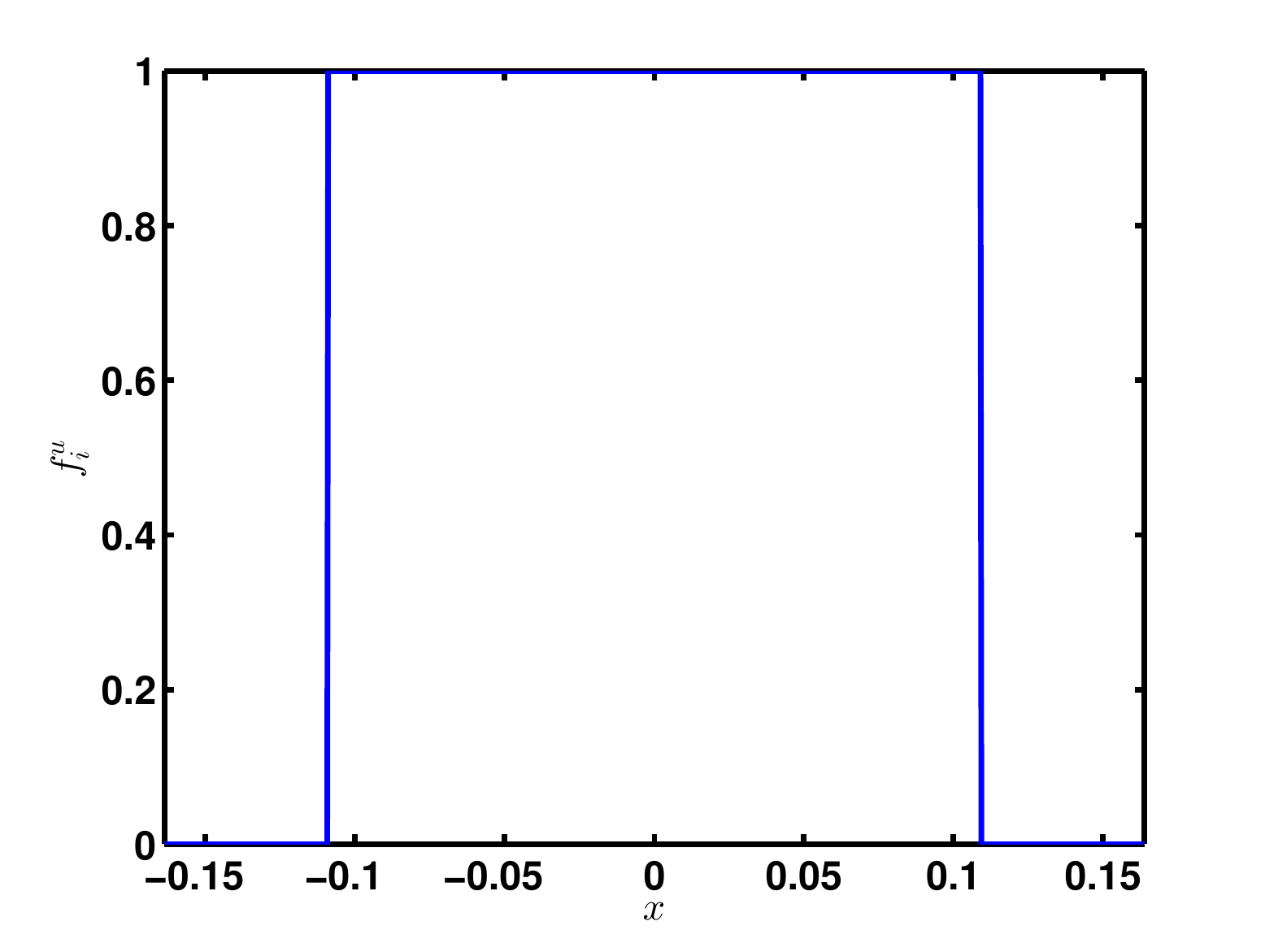}\hspace{0.01cm}
		\includegraphics[width=0.23\textwidth]{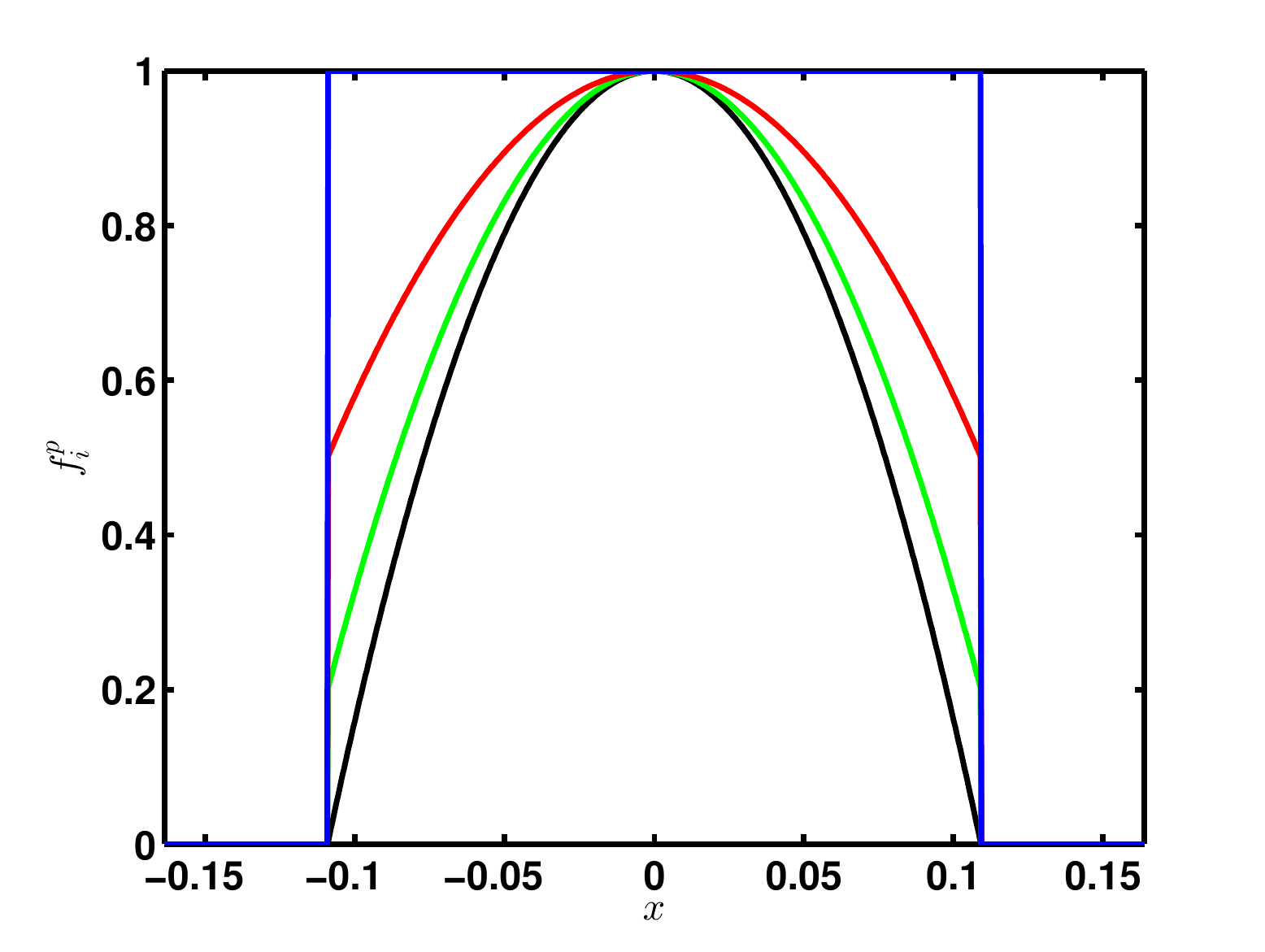}
		\caption{Coverage quality functions: $~f_i^u(q)$ [Left], $~f_i^p(q)$ [Right] for different values of $b$.
		%Blue for $b=1$, red for $b=0.5$, green for $b=0.8$ and black for $b=0$.
		}
		\label{fig:quality_functions}
	\end{figure}

	This choice of quality function simplifies the space partitioning since $\partial W_j \cap \partial W_i$ is either an arc of $\partial C_i$ if $z_i < z_j$ or of $\partial C_j$ if $z_i > z_j$. In the case where $z_i = z_j$, $\partial W_j \cap \partial W_i$ is chosen arbitrarily as the line segment defined by the two intersection points of $\partial C_i$ and $\partial C_j$. The resulting cells consist of circular arcs and line segments.
	
	If the sensing disk of a node $i$ is contained within the sensing disk of another node $j$, i.e. $C_i^s \cap C_j^s = C_i^s$, then $W_i = C_i^s$ and $W_j = C_j^s \setminus C_i^s$. An example partitioning with all of the above cases shown can be seen in Figure \ref{fig:partitioning_comparison} [Left], where the boundaries of the sensing disks $\partial C_i^s$ are in red and the boundaries of the cells $\partial W_i$ in black. Nodes $1$ and $2$ are at the same altitude so the arbitrary partitioning scheme is used. The sensing disk of node $3$ contains the sensing disk of node $4$ and nodes $5, 6$ and $7$ illustrate the general case.
	
	 The control law is also significantly simplified. Since $\frac{\partial f_i(q)}{\partial q_i} = 0$, the corresponding integral of the control law (\ref{control_law}) is 0. Additionally, $\upsilon_i^i$ and $\nu_i^i ~n_i$ are evaluated as
	\begin{eqnarray}
	\upsilon_i^i &=& \left \{
	\begin{aligned}
		\mathbb{I}_2, ~ z_i \leq z_j\\
		\textbf{0}_2, ~ z_i > z_j
	\end{aligned}
	\right.\\
	\nu_i^i \cdot n_i &=& \left \{
	\begin{aligned}
		\tan(a), ~ z_i \leq z_j\\
		0, ~ z_i > z_j
	\end{aligned}
	\right.
	\label{jacobian_evaluation}
	\end{eqnarray}
	Because both $\upsilon_i^i$ and $\nu_i^i \cdot n_i$ are zero when ${z_i > z_j \Rightarrow f_i^u < f_j^u}$, the integrals over ${\partial W_i \cap \partial W_j}$ are zero for both the planar and altitude control laws when ${f_i^u < f_j^u}$ and as a result need only be evaluated when ${f_i^u > f_j^u}$.
	
	The control law essentially maximizes the volume contained by all the cylinders defined by $f_i^u, ~i \in I_n$, under the constraints imposed by the network and area of interest.
	
	While the control law in the general case guarantees that $z_i < z^{\max}, \forall i \in I_n$, this particular one also guarantees that $z_i > z^{\min}, \forall i \in I_n$. If $z_i = z^{\min}$, then it follows that $\frac{\partial f_i(q)}{\partial z_i} = 0$, resulting in the integral over $W_i$ being zero. Since $f_i^u$ is positive and only arcs of $\partial W_i$ where $f_i^u > f_j^u$ are integrated over, the remaining part of the control law is positive, resulting in the node moving to a higher altitude.

	\subsection{Decreasing coverage quality}
	A more realistic assumption for some types of visual sensors would be for the coverage quality to decrease as distance from the node increases, to take into account lens distortion for example. The coverage quality of a point $q \in C_i^s$ is maximum directly below the node and decreases as ${\parallel q - q_i \parallel}$ increases. One such function is an inverted paraboloid whose maximum value depends on $z_i$ and could be defined as
	\small
	\begin{equation*}
	f_i^p(q) = \left \{
	\begin{aligned}
		&~\left[ 1 - \frac{1-b}{\left[ z_i \tan(a) \right]^2} \left[  \left( x - x_i\right)^2 + \left( y - y_i\right)^2 \right] \right] f_i^u, & ~q \in C_i^s\\
		&~0, & ~q \notin C_i^s
	\end{aligned}
	\right.\\
	\end{equation*}
	\normalsize
	where $b \in (0,1)$ is the coverage quality on $\partial C_i$ as a percentage of the coverage quality on $q_i$ and $f_i^u$ is the uniform coverage quality function defined previously and is used to set the maximum value of the paraboloid. A projection of this function on the $y-z$ plane can be seen in the right portion of Figure \ref{fig:quality_functions} for a node at $X_i = \left[ 0,0, z^{\max} \right]^T$ and various $b$-values; $b \in \left\{1,~0.5,~0.2,~0\right\}$ corresponds to blue, red, green and black colored curves, respectively. It can be shown that $f_i^p(q) = 0$ when $\parallel q-q_i \parallel > z_i ~\tan{a}$, i.e. when $q$ is outside the sensing disk. Additionally, when $\parallel q-q_i \parallel = z_i ~\tan{a}$, $f_i^p = b f_i^u$. 
	
	For the space partitioning it is necessary to solve $f_i^p(q) \geq f_j^p(q)$ in order to find part of $\partial W_j \cap W_i$. The resulting solution on the $x - y$ plane is a disk $C_{i,j}^i$ when $z_i \neq z_j$ and the halfplane defined by the perpendicular bisector of $q_i$ and $q_j$ when $z_i = z_j$. Finally, the cell of node $i$ with respect to node $j$ is $W_i = C_i^s \cap C_{i,j}^i$ if $z_i < z_j$ and $W_i = C_i^s \setminus \left( C_j^s \cap C_{i,j}^i \right)$ if $z_i > z_j$.  Figure \ref{fig:cell_creation} shows the intersection of two paraboloids $f_i^p$ and $f_j^p$ with $z_i < z_j$ and the boundaries of the resulting cells above them. The red line shows $\partial C_{i,j}^i \cap \left( C_i^s \cap C_j^s \right)$, i.e. the boundary of the intersection disk constrained within the sensing disks of the two nodes, since only for $q \in C_i^s \cap C_j^s$ can $f_i^p(q) = f_j^p(q)$. The green arcs represent $\partial W_i \cap \partial C_i^s \cap W_j$ while the blue ones represent $\partial W_i \cap \partial \mathcal{O}$. Thus the boundary of the cell of any node $i$ can be decomposed into disjoint sets as 
	\small
	\begin{equation*}
	\partial W_i = \left\{ (\partial W_i \cap \partial \mathcal{O}) \bigcup_{j \in N_i} \left( (\partial W_i \cap \partial C_{i,j}^i) \cap (\partial W_i \setminus \partial C_{i,j}^i \cap W_j) \right) \right\}
	\end{equation*}.
	\normalsize
		
	\begin{figure}[htb]
		\centering
		\includegraphics[width=0.2\textwidth]{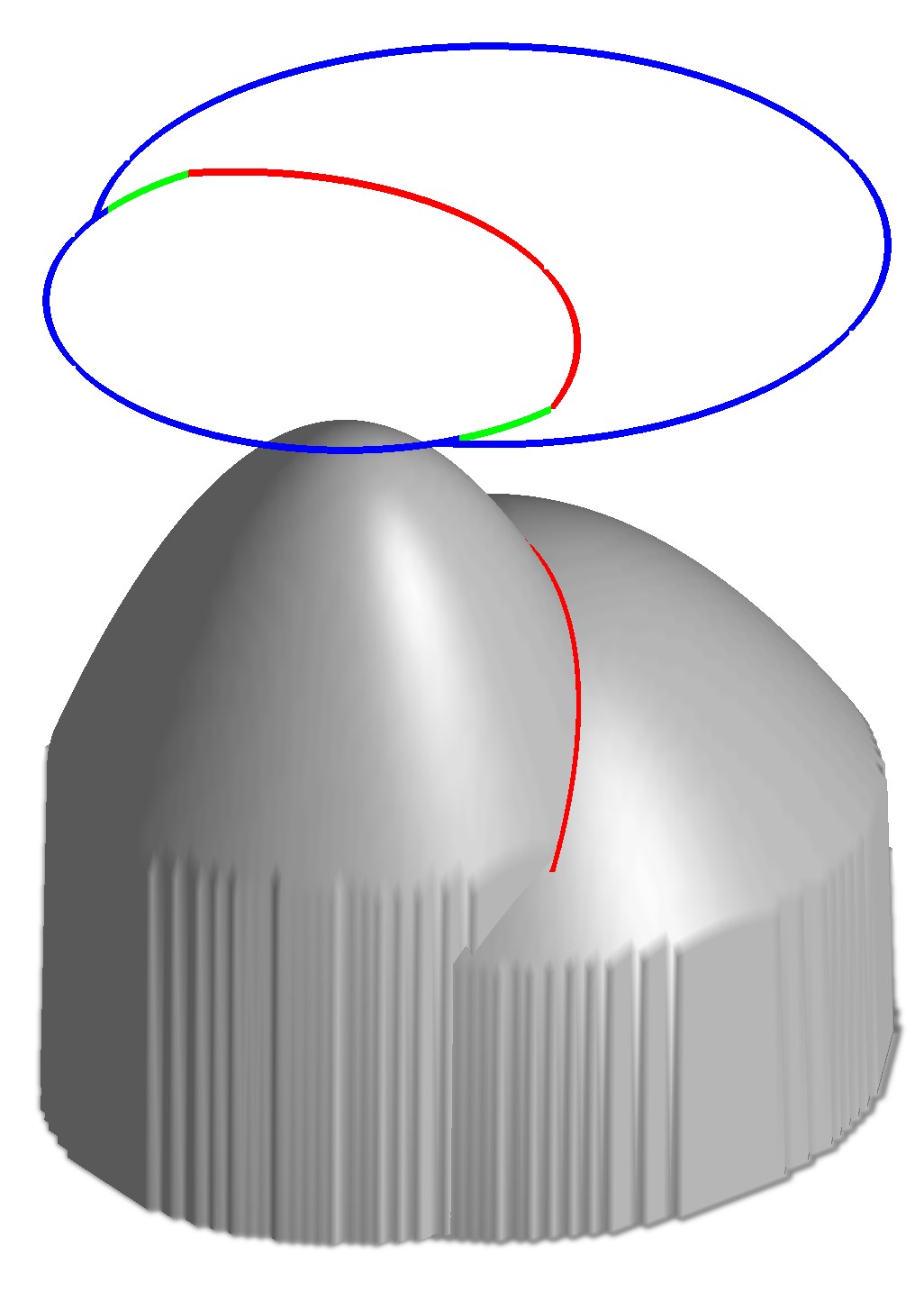}
		\caption{$\partial W_i$-decomposition when $f_i^p$ is used.}
		\label{fig:cell_creation}
	\end{figure}

	If the sensing disk of a node $i$ is contained within the sensing disk of another node $j$, i.e. $C_i^s \cap C_j^s = C_i^s$, then $W_i \subset C_i^s$ and $W_j = C_j^s \setminus W_i$. It is important to note that in the general case, $\partial W_j \cap \partial W_i$ consists of circular arcs from both $\partial C_{i,j}^i$ and $\partial C_i^s$ or $\partial C_j^s$ (depending on the relationship of $z_i$ and $z_j$). 
	
	An example partitioning with all of the above cases shown can be seen in Figure \ref{fig:partitioning_comparison} [Right], where the boundaries of the sensing disks $\partial C_i^s$ are in red and the cells $W_i$ in black. The sensing disk of node $3$ contains the sensing disk of node $4$ and nodes $5, 6$ and $7$ illustrate the general case.
	
	For the evaluation of $\upsilon_i^i$ and $\nu_i^i \cdot n_i$ the partitioning of $\partial W_i$ explained previously will be used.
	Since $\forall q \in \partial C_{i,j}^i \Rightarrow f_i^p(q) = f_j^p(q)$, the integral over $\partial C_{i,j}^i$ is zero. For the remaining part of $\partial W_i \cap W_j$ which is $\partial W_i \setminus \partial C_{i,j}^i \cap W_j$, the evaluation of $\upsilon_i^i$ and $\nu_i^i \cdot n_i$ is the same as in the case of the uniform coverage function $f_i^u$. For the same reasons as in the uniform coverage quality, the integrals over ${\partial W_i \cap \partial W_j}$ are zero for both the planar and altitude control laws when ${f_i^p < f_j^p}$ and as a result need only be evaluated when ${f_i^p > f_j^p}$.
	
	The control law essentially leads to the maximization of the volume contained by all the paraboloid-like functions $f_i^p, ~i \in I_n$, under the constraints imposed by the network and area of interest.
	
	This control law also guarantees that that $z_i > z^{\min}, \forall i \in I_n$ in addition to $z_i < z^{\max}, \forall i \in I_n$. Because $f_i^p$ is positive and only arcs of $\partial W_i$ where $f_i^p > f_j^p$ are integrated over, the integral over $\partial W_i$ is always positive when $z_i = z^{\min}$. It can also be shown that $\frac{\partial f_p(q)}{\partial z_i} > 0, \forall q \in W_i$ when $z_i = z^{\min}$ and thus the integral over $W_i$ of $u_{i,z}$ will also be positive. As a result the altitude control law is always positive when $z_i = z^{\min}$ resulting in the altitude increase of node $i$.
	
%%%%%%%%%%%%%%%%%%%%%%%%%%%%%%%%%%%%%%%%%
\section{Simulation Studies}
\label{section:simulation_studies}
	
	Simulation results of the proposed control law using the uniform coverage quality function $f_i^u$ are presented in this section. The region of interest $\Omega$ is the same as the one used in \cite{Stergiopoulos_IETCTA10} for consistency. All nodes are identical with a half sensing cone angle $a = 20^{\circ}$ and $z_i \in [0.3, ~2.3], ~\forall i \in I_n$. The boundaries of the nodes' cells are shown in black and the boundaries of their sensing disks in red.
	
	\begin{remark}
	When using $f_i^u(q)$ it is possible to observe jittering on the cells of some nodes $i$ and $j$. This can happen when $z_i = z_j$ and the arbitrary boundary $\partial W_i \cap \partial W_j$ is used. Once the altitude of one of the nodes changes slightly, the boundary between the cells will change instantaneously from a line segment to a circular arc. The coverage-quality objective $\mathcal{H}$ however will present no discontinuity when this happens.
	\end{remark}
	
	\subsection{Case Study I}
	In this simulation three nodes start grouped as seen in Figure (\ref{fig:uniform_3_nodes_2D}) [Left]. Since the region of interest is large enough for three optimal disks $C_{i,opt}^s$ to fit inside, all the nodes converge at the optimal altitude $z^{opt}$. As it can be seen in Figure~\ref{fig:uniform_3_nodes_area}, the area covered by the network is equal to $\mathcal{A}\left( \bigcup_{i \in I_n} C_i^s \right)$ and the coverage-quality criterion has reached $\mathcal{H}_{opt}$ = $\mathcal{A}\left( \bigcup_{i \in I_n} C_{i,opt}^s \right)$. However since all nodes converged at $z^{opt}$, the addition of more nodes will result in significantly better performance coverage and quality wise, as is clear from Figure \ref{fig:uniform_3_nodes_2D} [Right] and Figure \ref{fig:uniform_3_nodes_area} [Left]. Figure \ref{fig:uniform_3_nodes_3D} shows a graphical representation of the coverage quality at the initial and final stages of the simulation. It is essentially a plot of all $f_i^u(q)$ inside the region of interest. The volume of the cylinders in Figure \ref{fig:uniform_3_nodes_3D} [Right] is the maximum possible. The trajectories of the UAVs in $\mathbb{R}^3$ can be seen in Figure \ref{fig:uniform_3_nodes_traj} in blue and their projections on the region of interest in black. 

	% % % % % % % % % % % % FIGURES 3 NODES % % % % % % % % % % % %
	\begin{figure}[tp]
		\centering
		\includegraphics[width=0.23\textwidth]{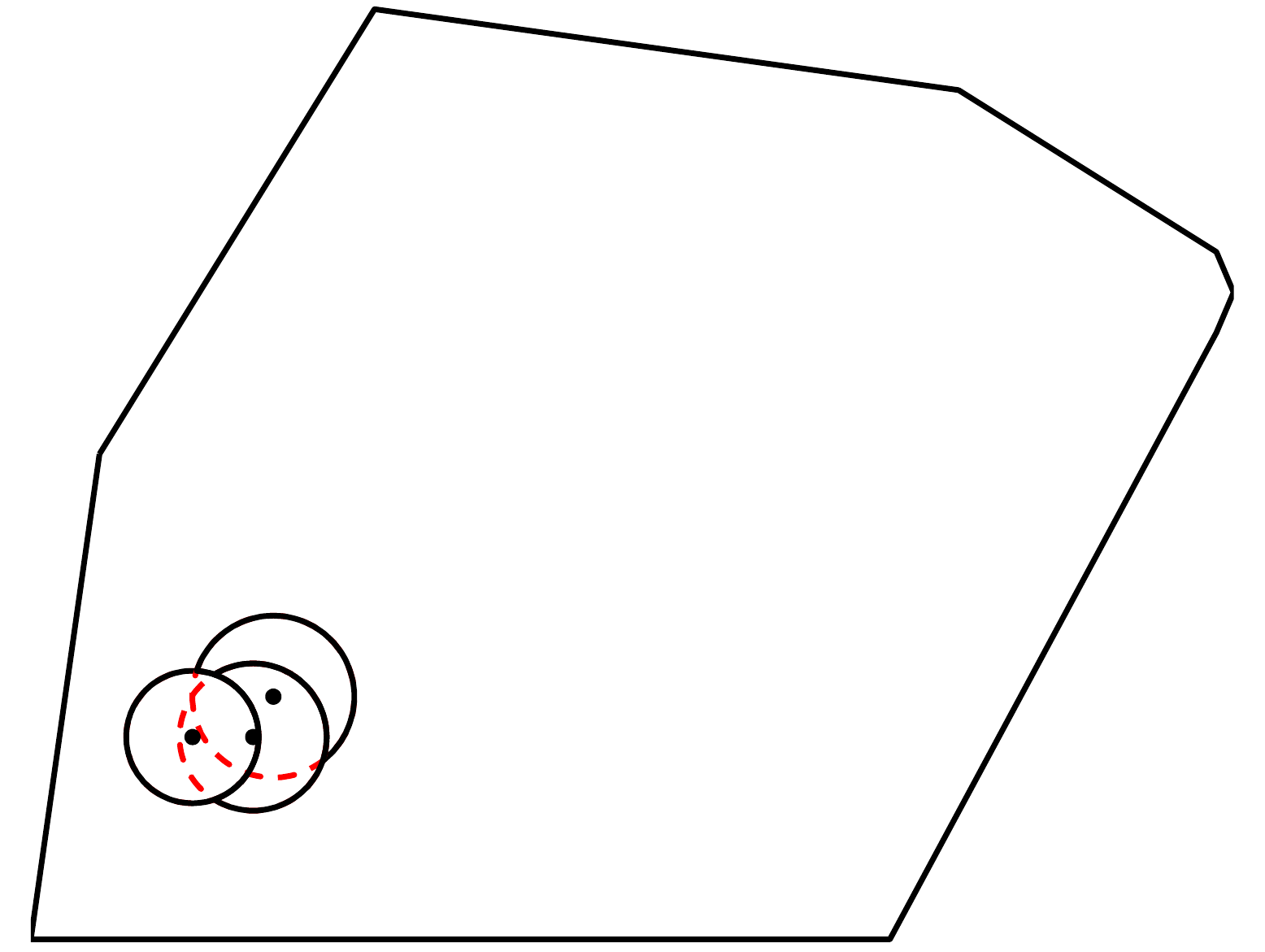}\hspace{0.01cm}
		\includegraphics[width=0.23\textwidth]{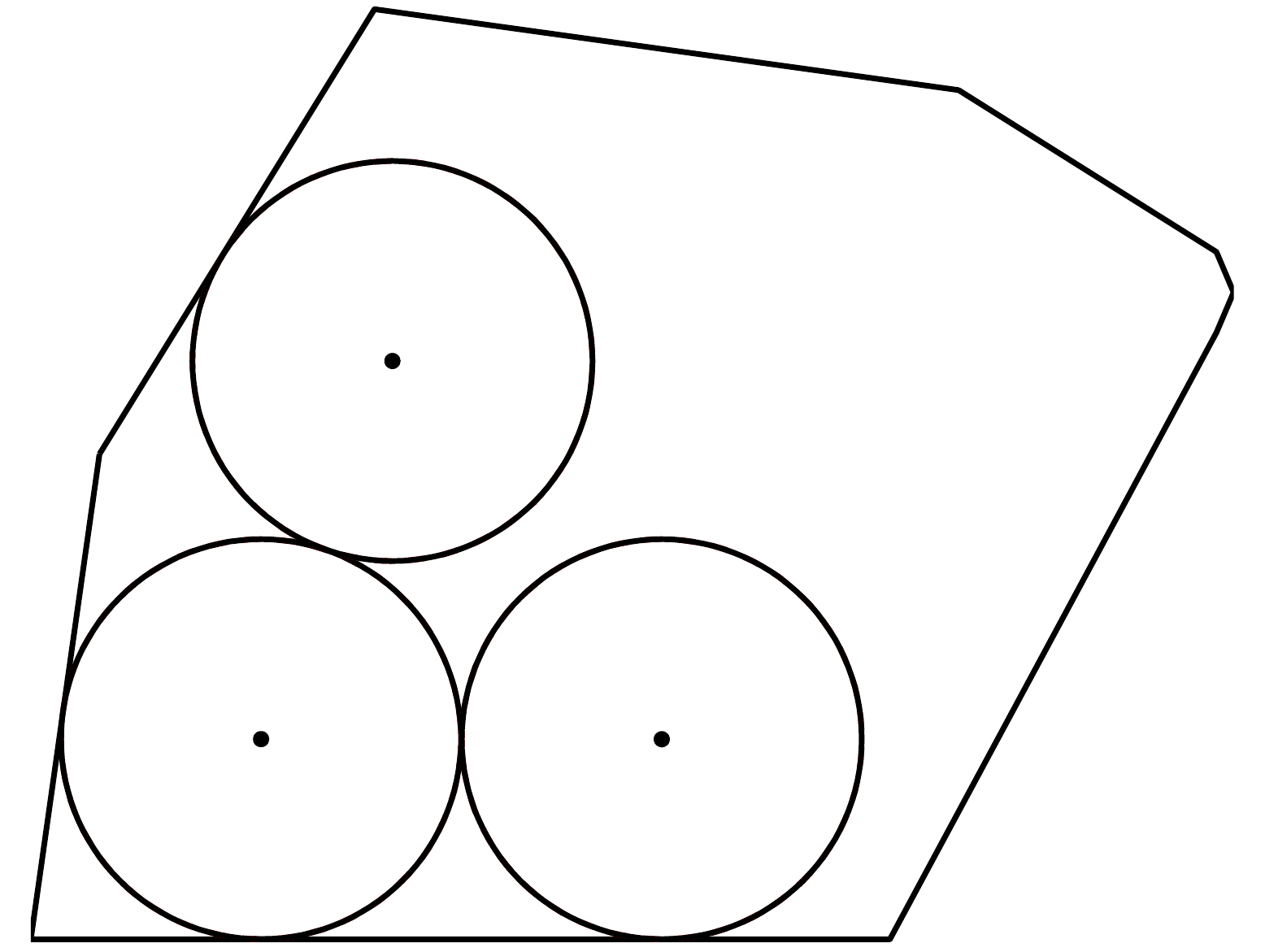}
		\caption{Initial [Left] and final [Right] network configuration and space partitioning.}
		\label{fig:uniform_3_nodes_2D}
	\end{figure}
	
	\begin{figure}[htbp]
		\centering
		\includegraphics[width=0.23\textwidth]{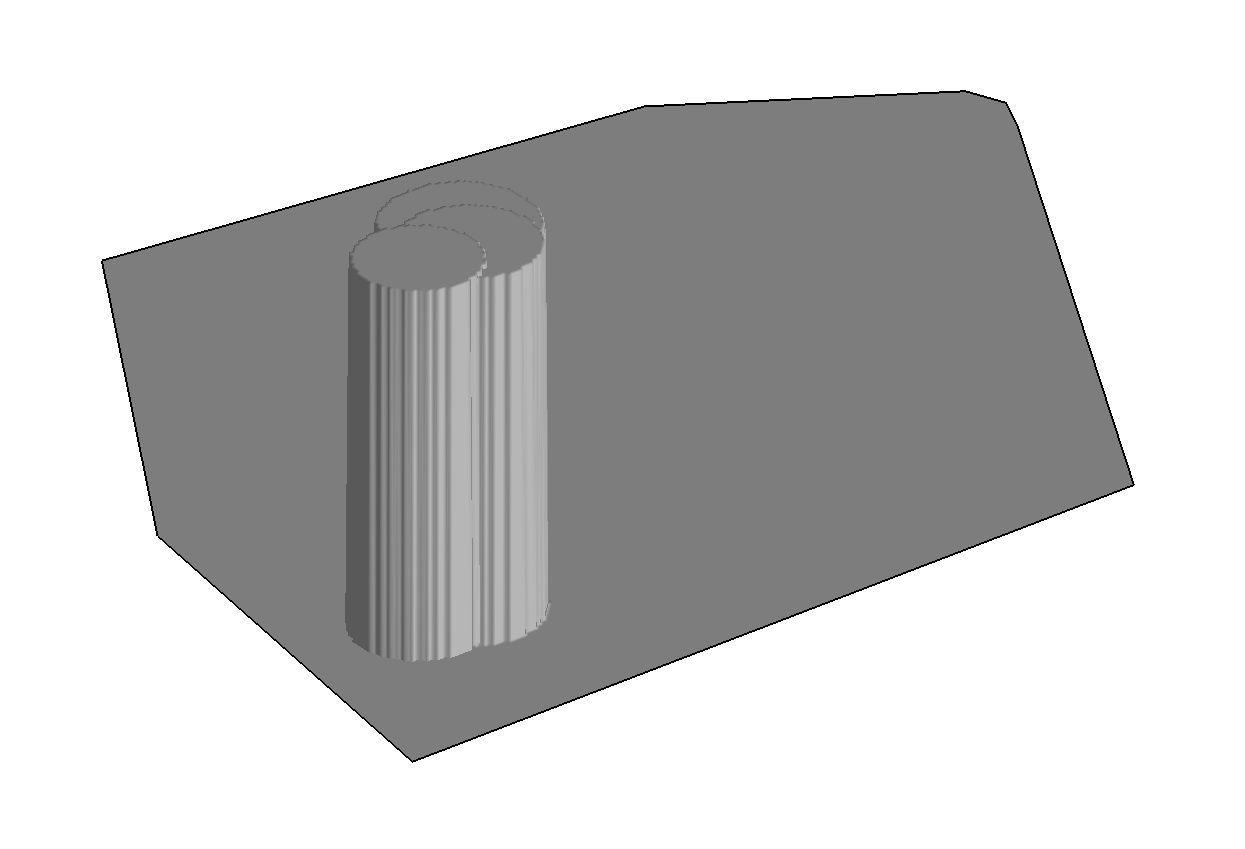}\hspace{0.01cm}
		\includegraphics[width=0.23\textwidth]{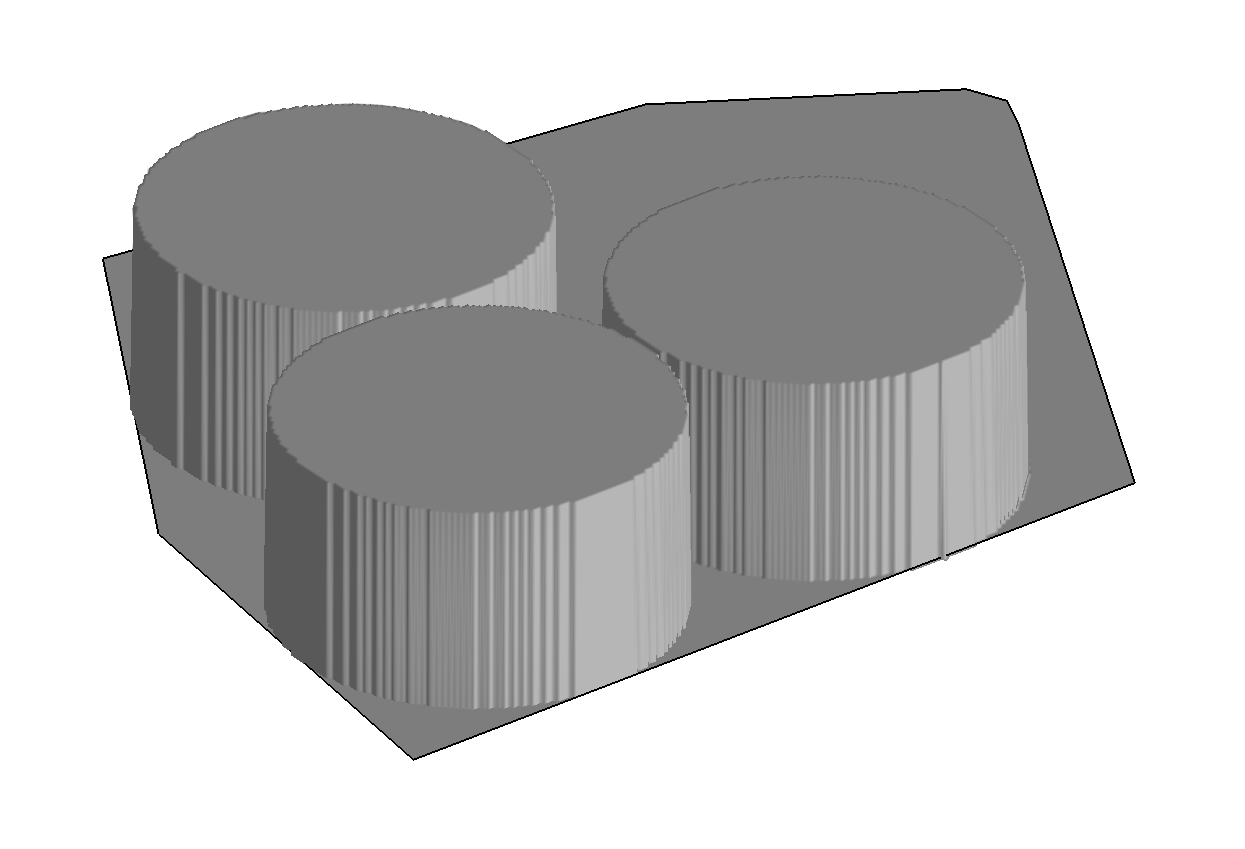}
		\caption{Initial [Left] and final [Right] coverage quality.}
		\label{fig:uniform_3_nodes_3D}
	\end{figure}
	
	\begin{figure}[htbp]
		\centering
		\includegraphics[width=0.4\textwidth]{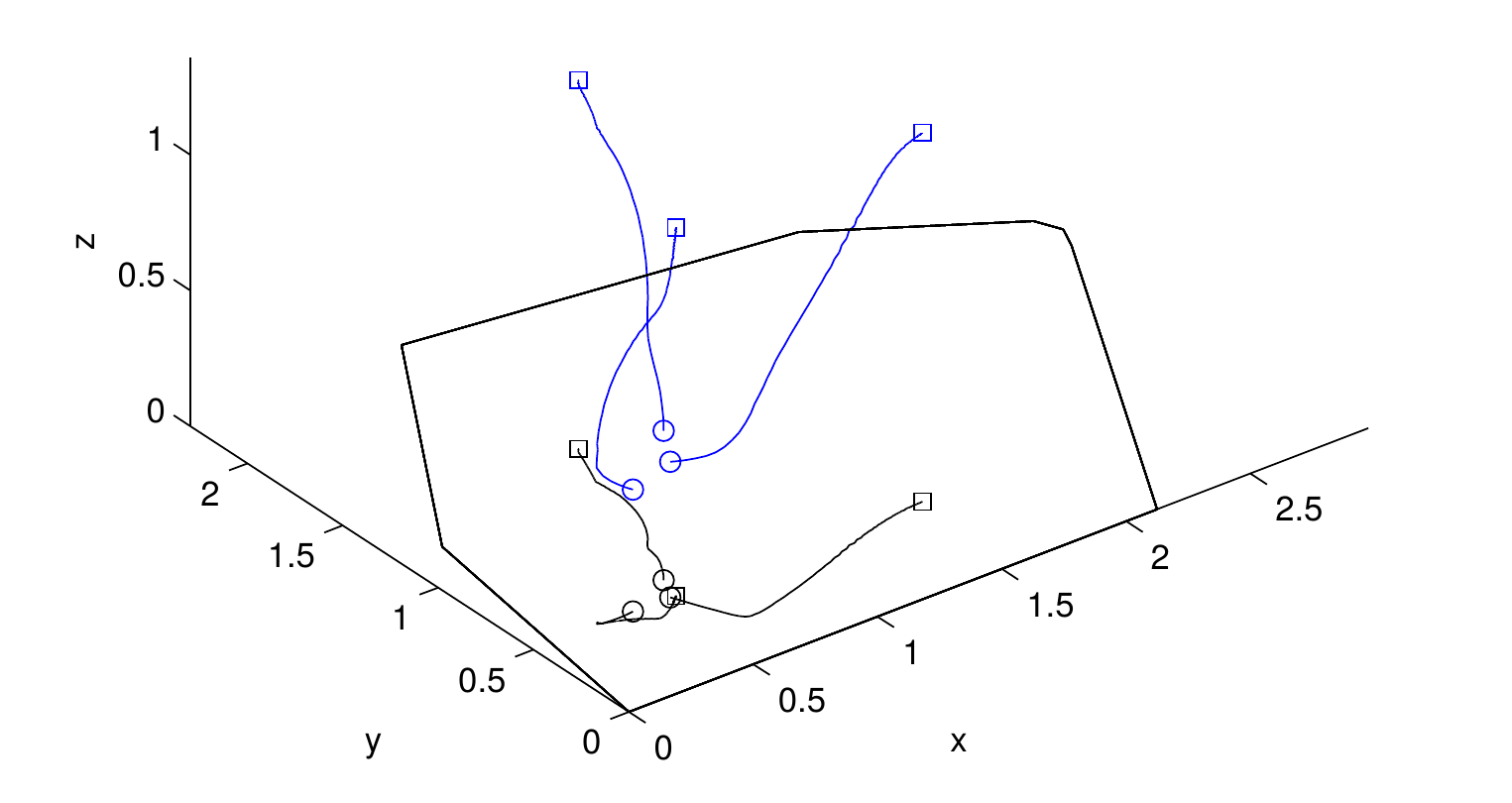}
		\caption{Node trajectories (blue) and their projections on the sensed region (black).}
		\label{fig:uniform_3_nodes_traj}
	\end{figure}
	
	\begin{figure}[htb]
		\centering
		\includegraphics[width=0.23\textwidth]{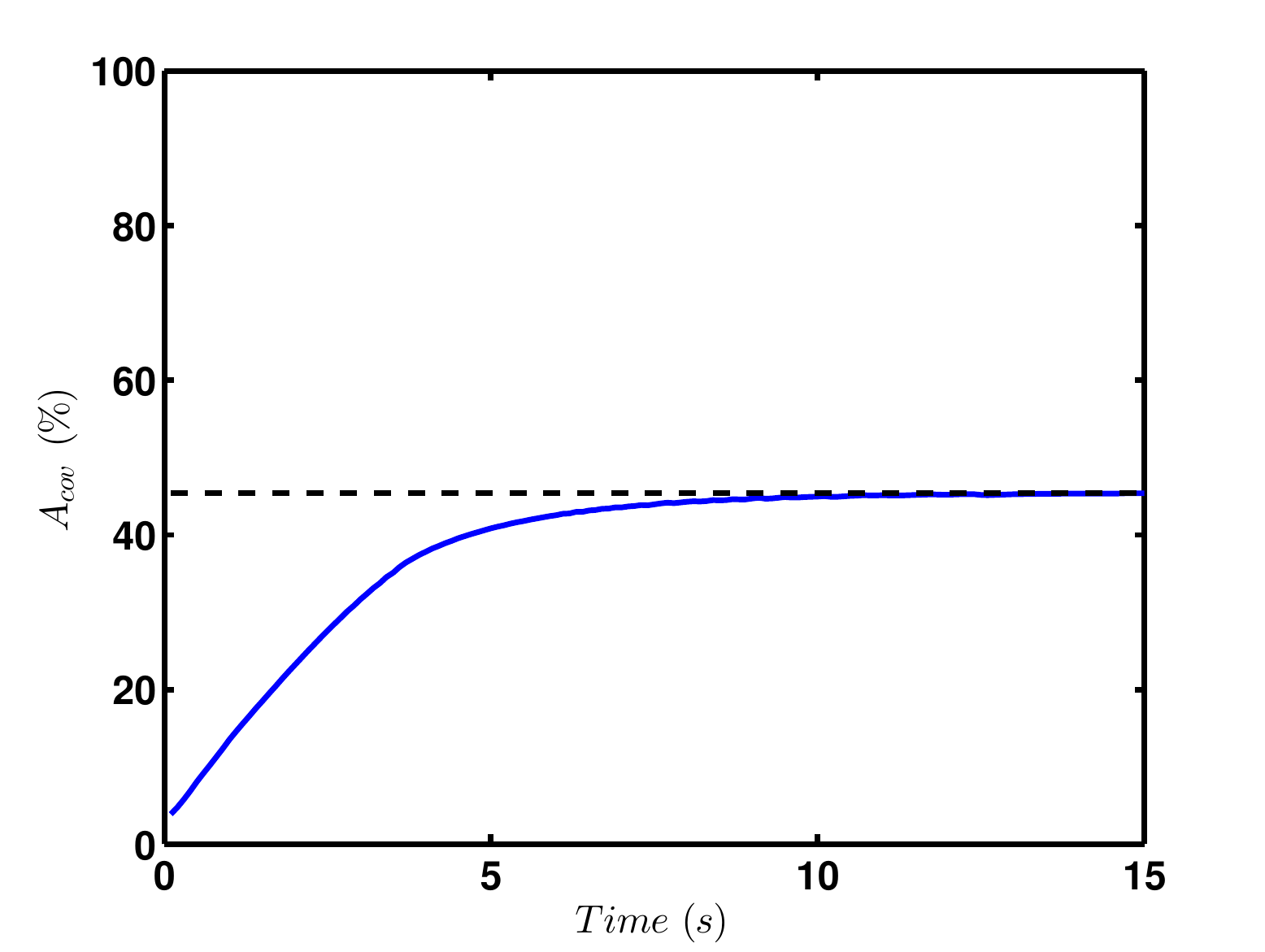}\hspace{0.01cm}
		\includegraphics[width=0.23\textwidth]{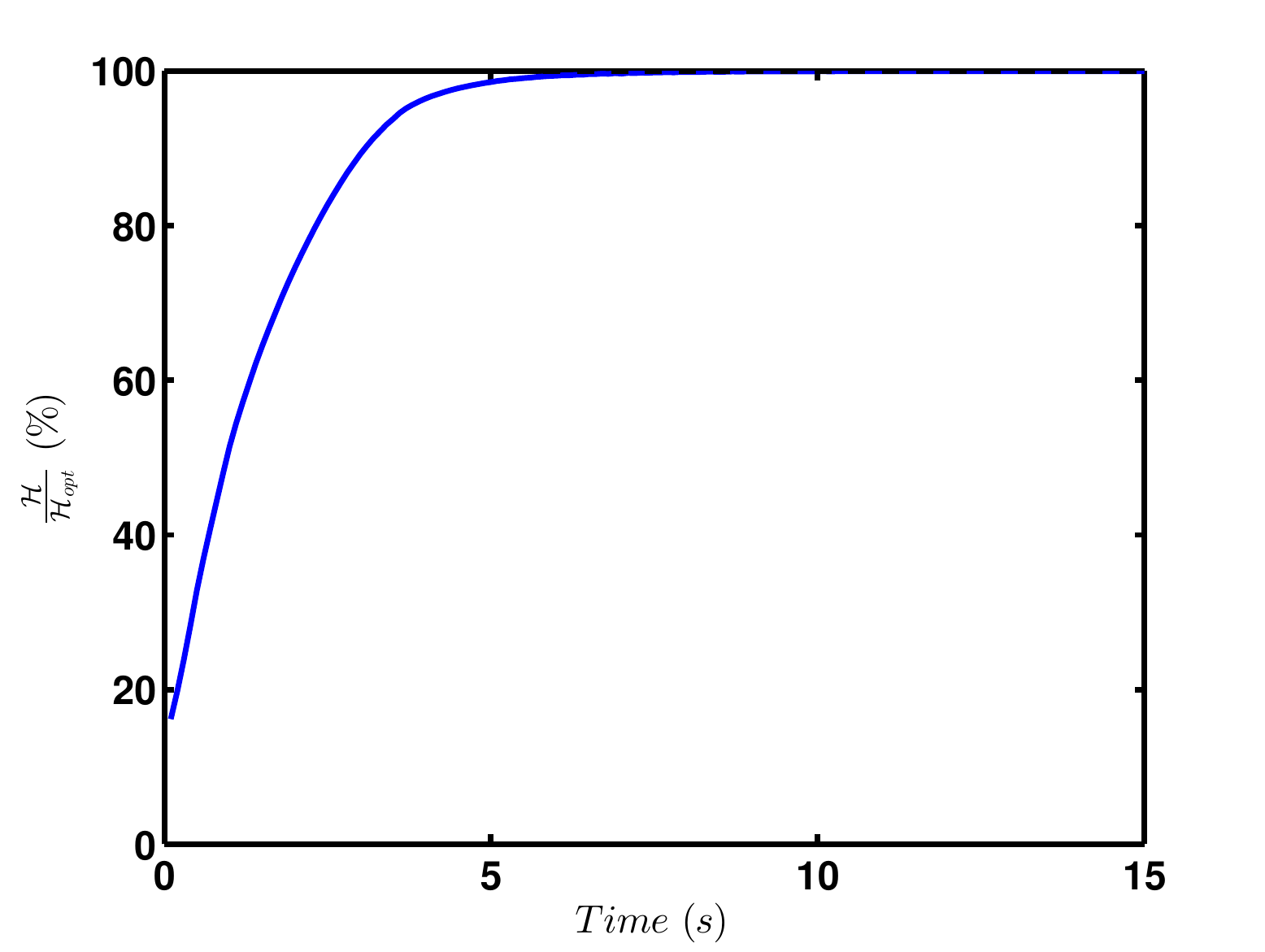}
		\caption{$\frac{\mathcal{A}\left( \bigcup_{i \in I_n} C_i^s \right)}{\mathcal{A}(\Omega)}$ [Left] and $\frac{\mathcal{H}}{\mathcal{H}_{opt}}$ [Right].}
		\label{fig:uniform_3_nodes_area}
	\end{figure}

	\subsection{Case Study II}
	A network of nine nodes, identical to those in Case Study I, is examined in this simulation with an initial configuration as seen in Figure \ref{fig:uniform_9_nodes_2D} [Left]. The region $\Omega$ is not large enough to contain these nine $C_{i,opt}^s$ disks and so the nodes converge at different altitudes below $z^{opt}$. This is why the covered area never reaches $\mathcal{A}\left( \bigcup_{i \in I_n} C_i^s \right)$, which is larger than $\mathcal{A}(\Omega)$ and why $\mathcal{H}$ never reaches $\mathcal{H}_{opt}$, as seen in Figure \ref{fig:uniform_9_nodes_area}. It can be clearly seen though from Figure \ref{fig:uniform_9_nodes_2D} [Right] and Figure \ref{fig:uniform_9_nodes_area} [Left] that the network covers a significant portion of $\Omega$ with better quality than Case Study I. The volume of the cylinders in Figure \ref{fig:uniform_9_nodes_3D} [Right] has reached a local optimum. The trajectories of the UAVs in $\mathbb{R}^3$ can be seen in Figure \ref{fig:uniform_9_nodes_traj} in blue and their projections on the region of interest in black. It can be seen from the trajectories that the altitude of some nodes was not constantly increasing. This is expected behavior since nodes at lower altitude will increase the stable altitude of nodes at higher altitude they share sensed regions with. Once they no longer share sensed regions, or share a smaller portion, the stable altitude of the upper node will decrease, leading to a decrease in their altitude.

	% % % % % % % % % % % % FIGURES 9 NODES % % % % % % % % % % % %
		\begin{figure}[htb]
			\centering
			\includegraphics[width=0.23\textwidth]{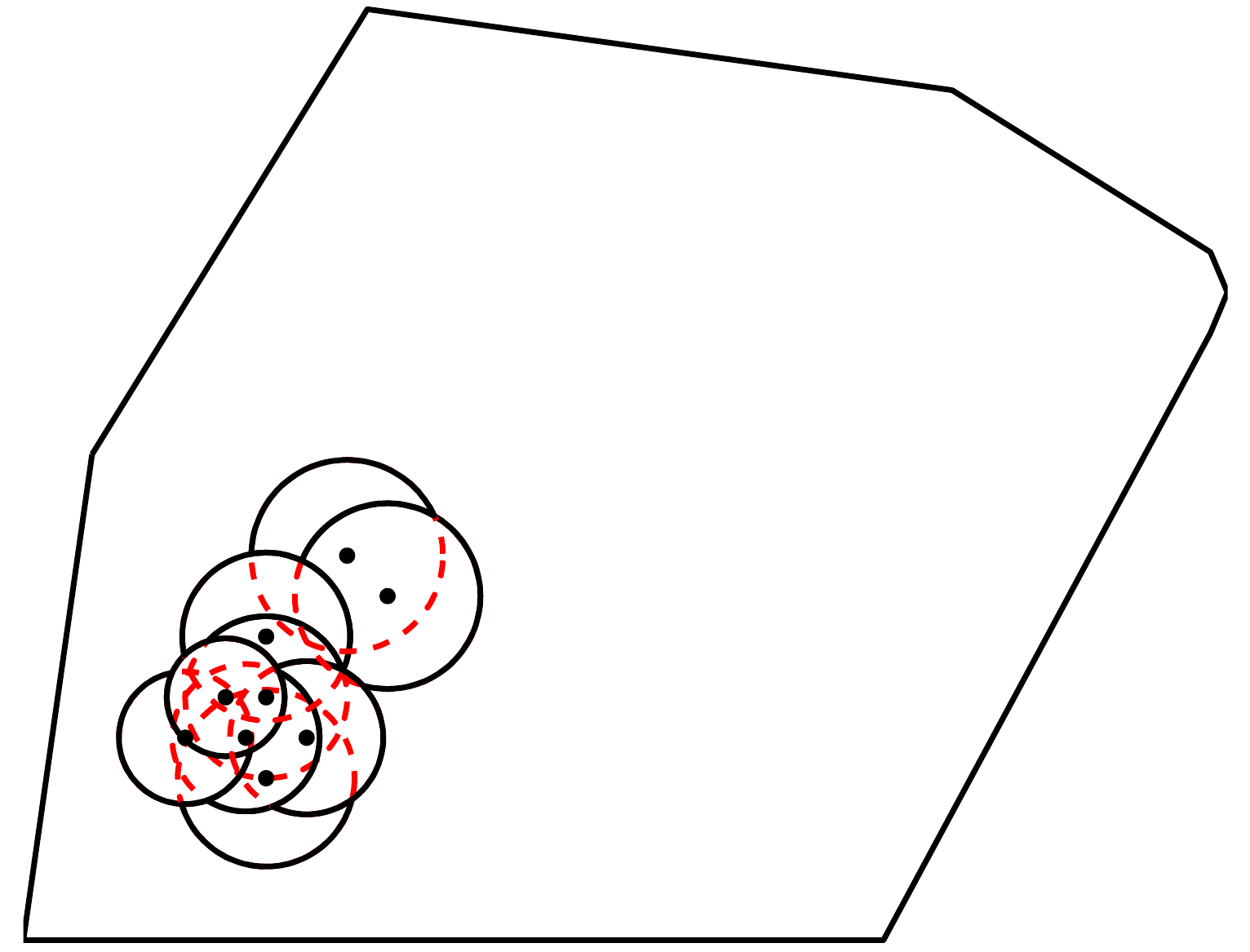}\hspace{0.01cm}
			\includegraphics[width=0.23\textwidth]{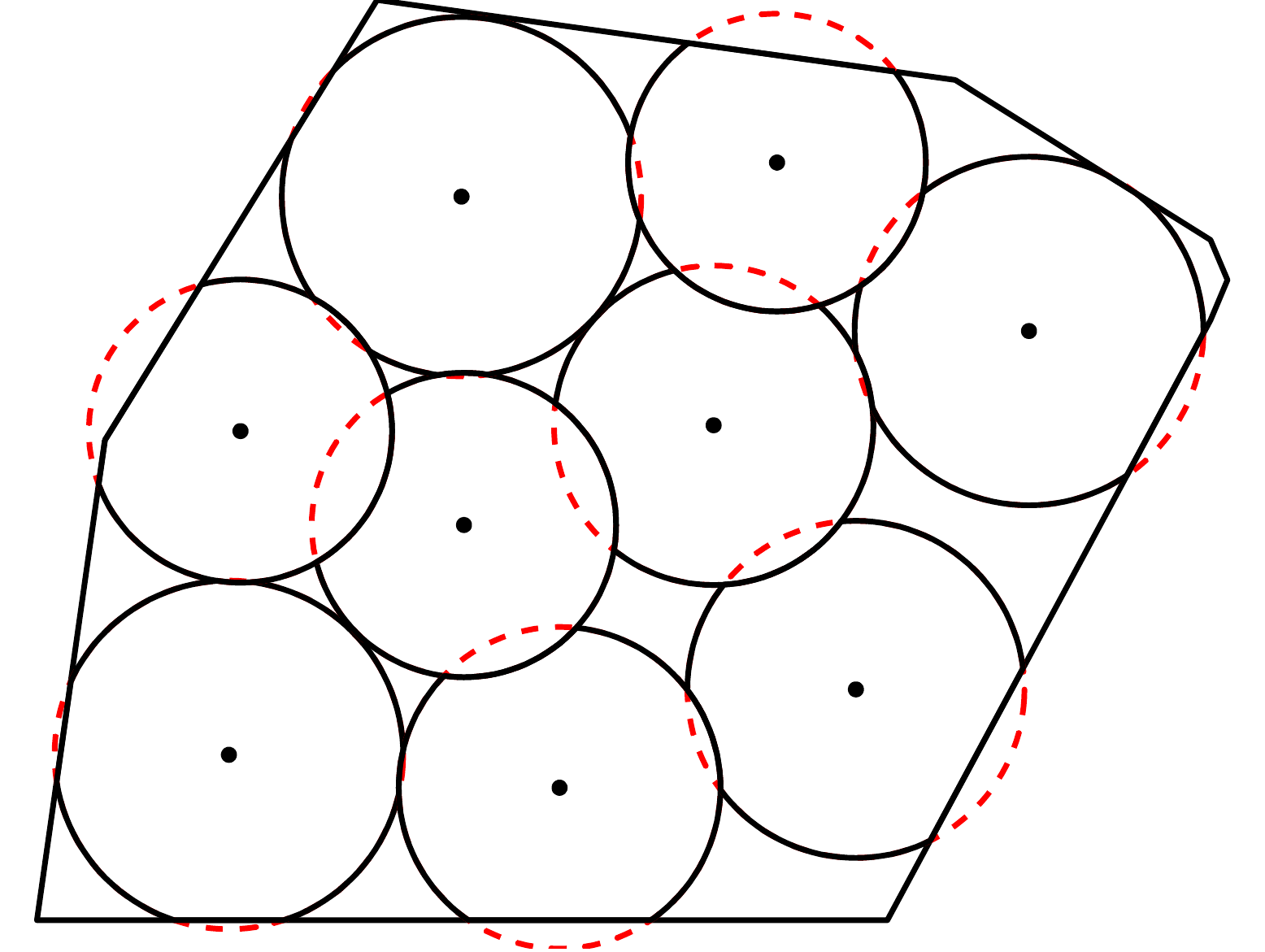}
			\caption{Initial [Left] and final [Right] network configuration and space partitioning.}
			\label{fig:uniform_9_nodes_2D}
		\end{figure}
		
		\begin{figure}[htb]
			\centering
			\includegraphics[width=0.23\textwidth]{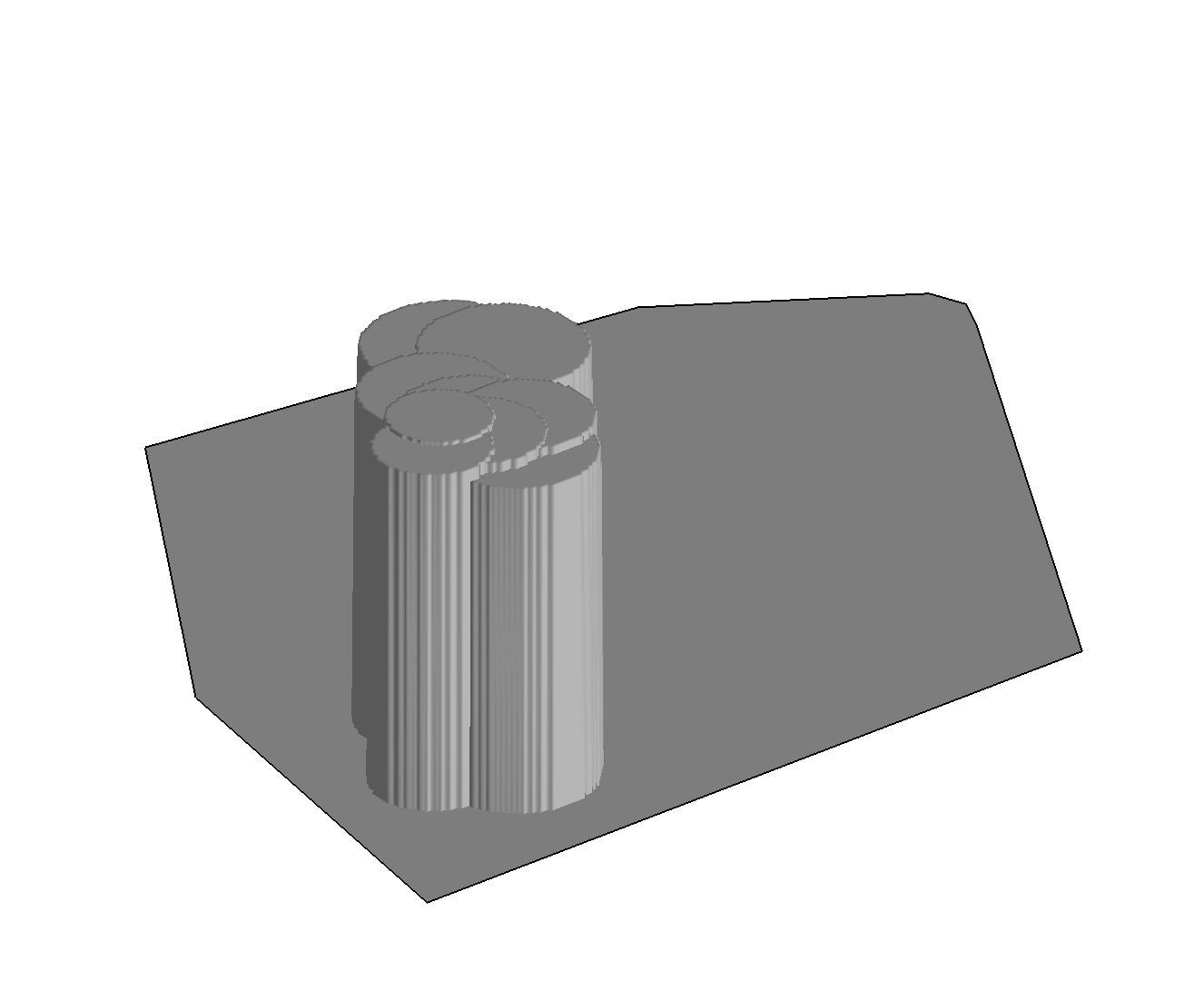}\hspace{0.01cm}
			\includegraphics[width=0.23\textwidth]{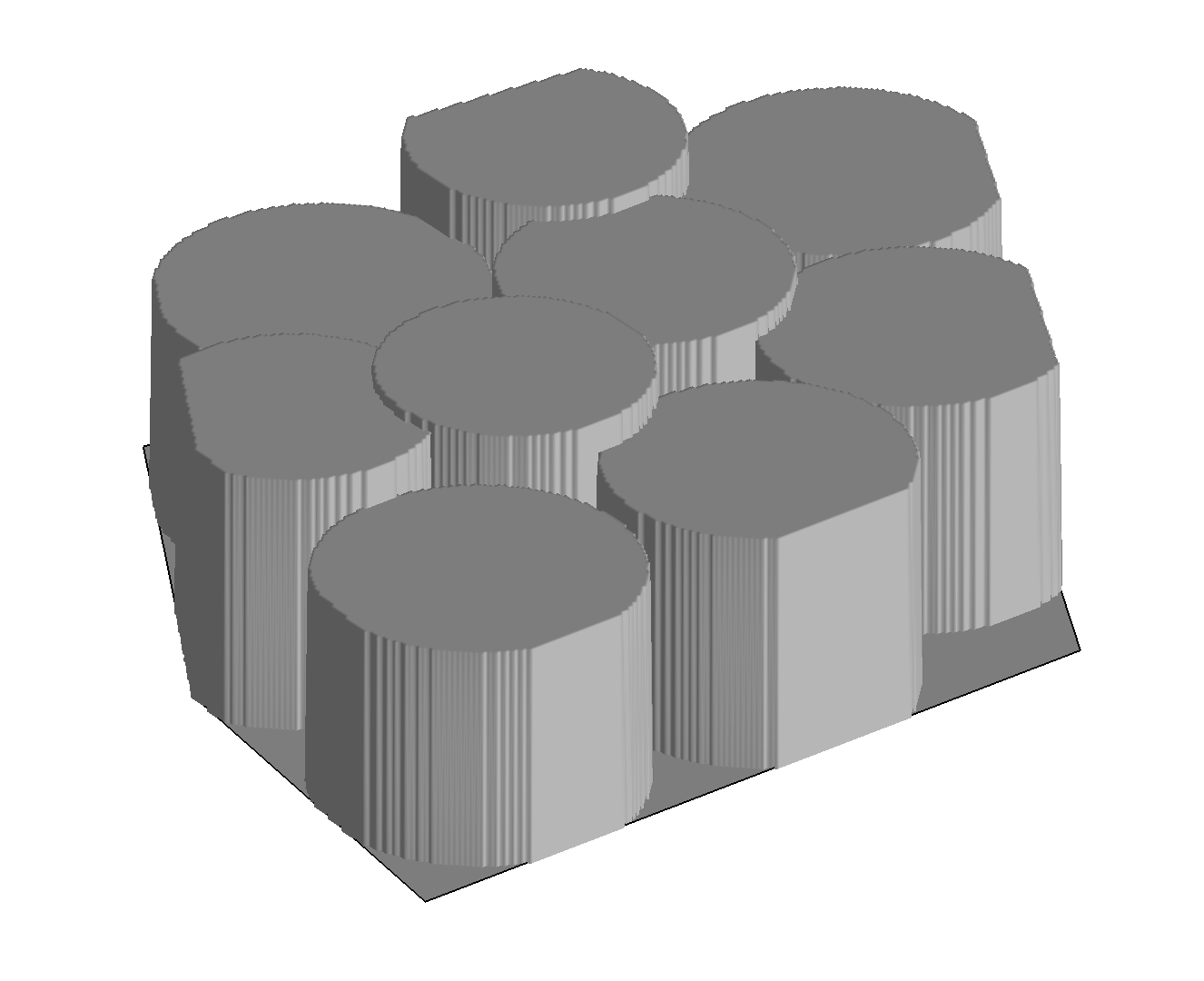}
			\caption{Initial [Left] and final [Right] coverage quality.}
			\label{fig:uniform_9_nodes_3D}
		\end{figure}
		
		\begin{figure}[htb]
			\centering
			\includegraphics[width=0.4\textwidth]{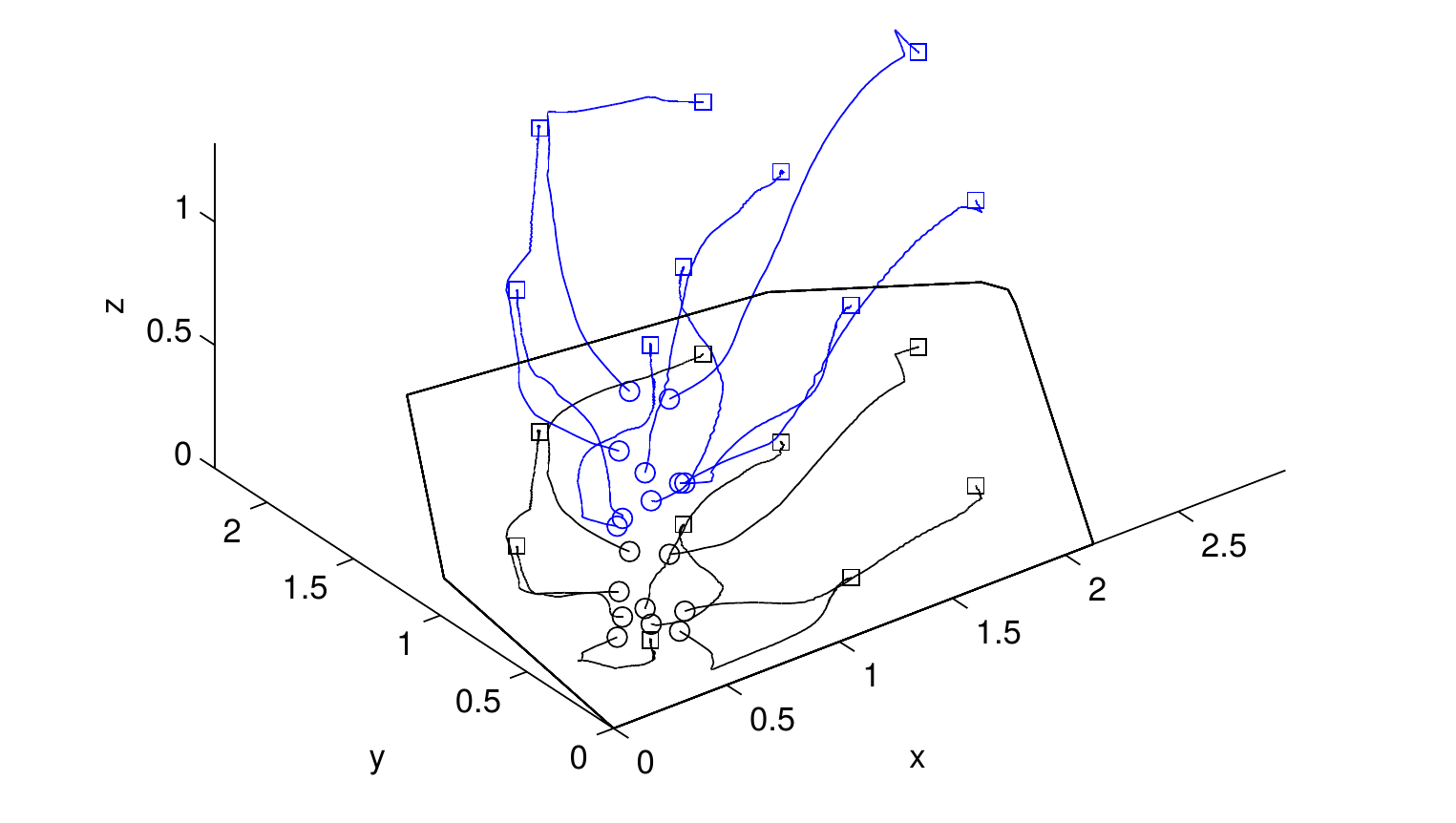}	
			\caption{Node trajectories (blue) and their projections on the sensed region (black).}
			\label{fig:uniform_9_nodes_traj}
		\end{figure}
		
		\begin{figure}[htb]
			\centering
			\includegraphics[width=0.23\textwidth]{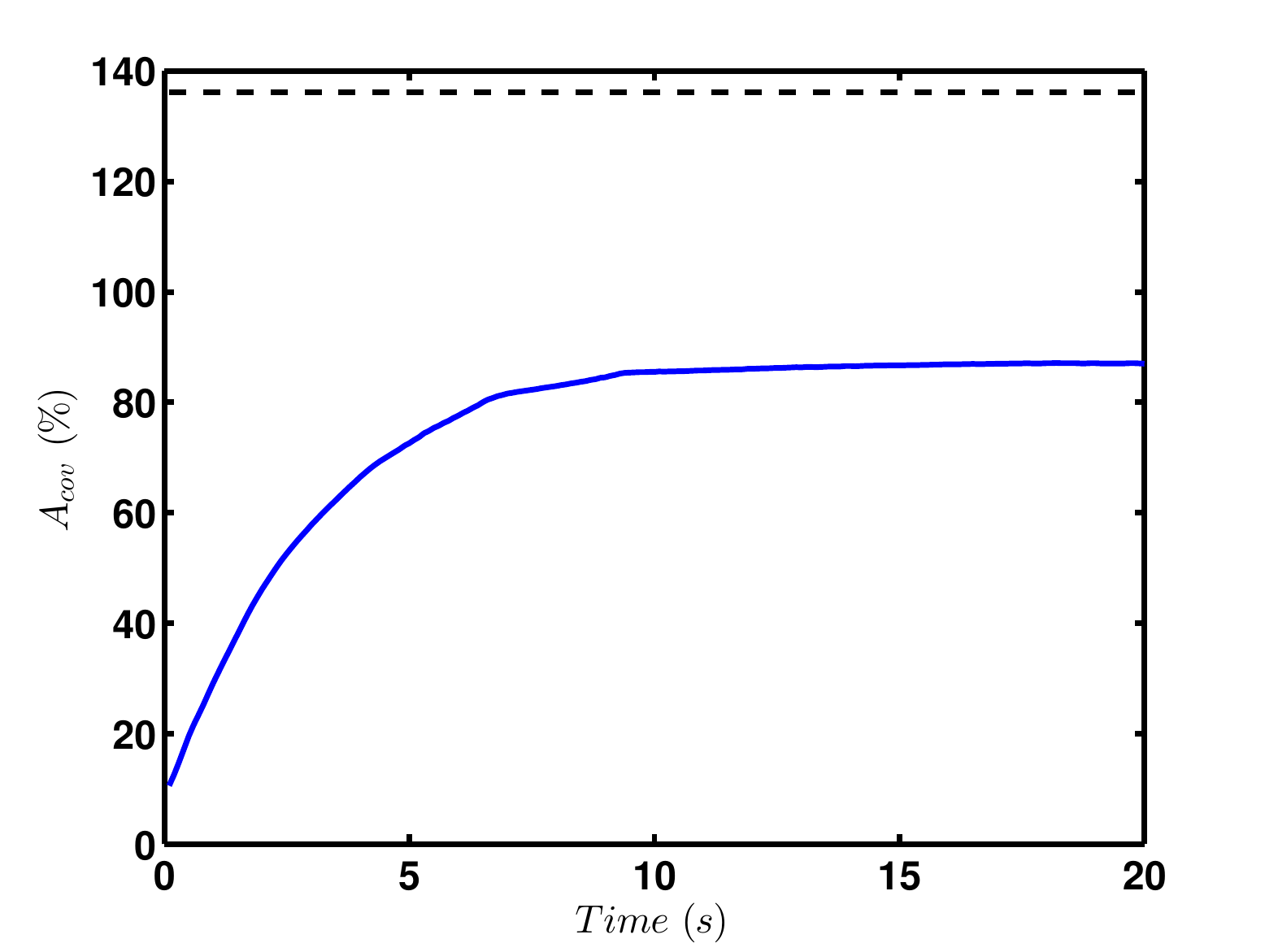}\hspace{0.01cm}
			\includegraphics[width=0.23\textwidth]{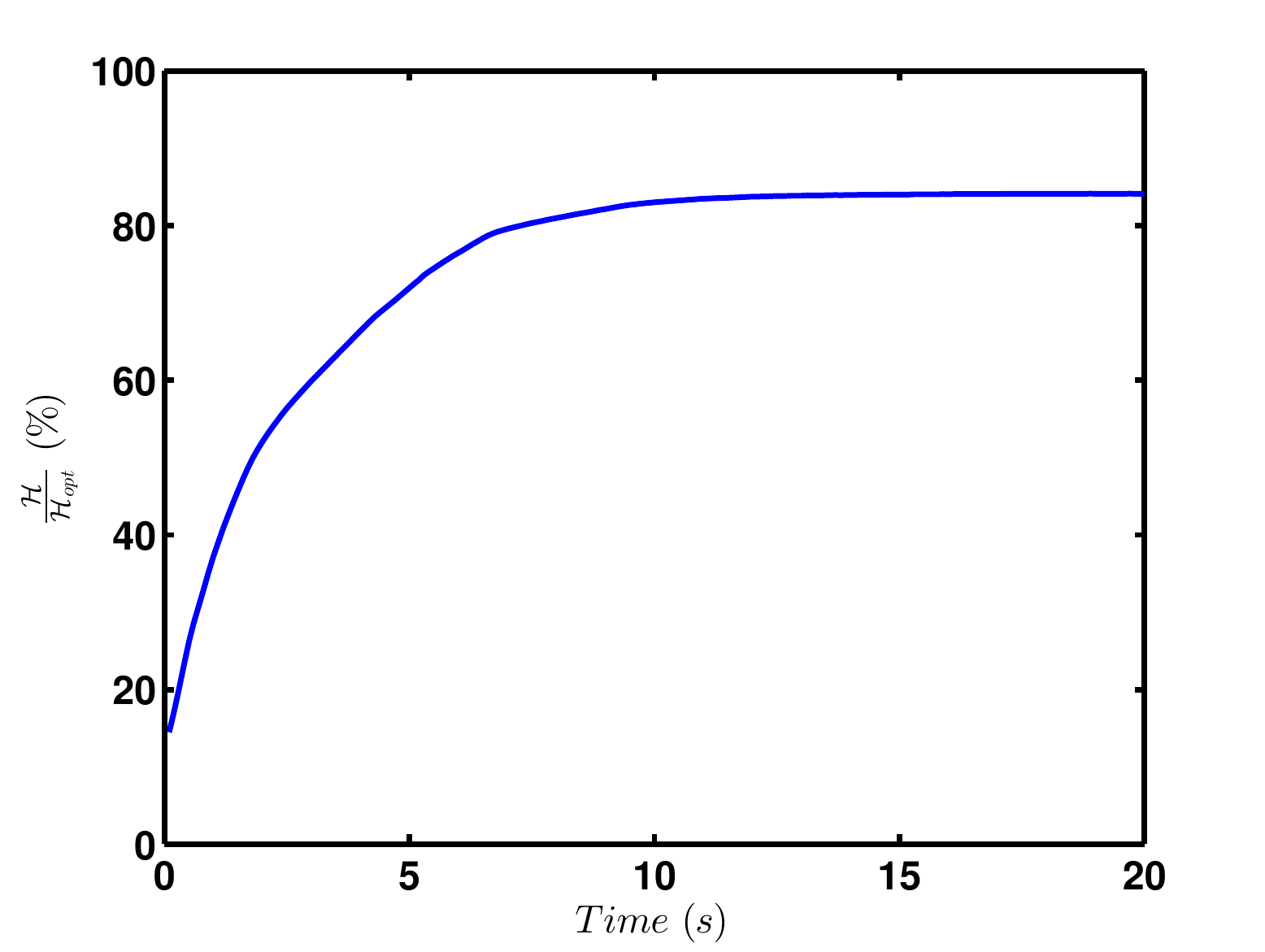}
		\caption{$\frac{\mathcal{A}\left( \bigcup_{i \in I_n} C_i^s \right)}{\mathcal{A}(\Omega)}$ [Left] and $\frac{\mathcal{H}}{\mathcal{H}_{opt}}$ [Right].}
			\label{fig:uniform_9_nodes_area}
		\end{figure}
	
%%%%%%%%%%%%%%%%%%%%%%%%%%%%%%%%%%%%%%%%%

\section{Conclusions}
Area coverage by a network of UAVs has been studied in this article by use of a combined coverage-quality metric. A partitioning scheme based on coverage quality is employed to assign each UAV an area of responsibility. The proposed control law leads the network to a locally optimal configuration which provides a compromise between covered area and coverage quality. It also guarantees that the altitude of all UAVs will remain within a predefined range, thus avoiding potential obstacles while also keeping the UAVs below their maximum operational altitude. Simulation studies are presented to indicate the efficiency of the proposed control algorithm.

\section*{APPENDIX}

The parametric equation of the boundary of the sensing disk $C_{i}^{s}(X_i,a)$ defined in (\ref{sensing}) is
\begin{equation*}
\gamma(k) \colon~ \left[\begin{array}{c} x \\ y \end{array}\right] = \left[\begin{array}{c} x_i + z_i \tan(a) ~\cos(k) \\ y_i + z_i \tan(a) ~\sin(k) \end{array}\right], ~k \in [0,2\pi)
\label{circle_parametric}
\end{equation*}

We will now evaluate $n_i$, $\upsilon_i^i(q)$ and $\nu_i^i(q)$ on $\partial W_i \cap \partial\mathcal{O}$ which is always an arc of the circle $\gamma(k)$.

The normal vector $n_i$ is given by
\begin{equation*}
n_i = \left[\begin{array}{c} \cos(k) \\ \sin(k) \end{array}\right], ~k \in [0,2\pi)
\end{equation*}

It can be easily shown that
\begin{equation*}
\upsilon_i^i(q) = \left[\begin{array}{cc} \frac{\partial x}{\partial x_i} & \frac{\partial x}{\partial y_i} \\ \frac{\partial y}{\partial x_i} & \frac{\partial y}{\partial y_i} \end{array}\right] = \mathbb{I}_2
\end{equation*}
and similarly that
\begin{equation*}
\nu_i^i(q) = \left[\begin{array}{c} \frac{\partial x}{\partial z_i} \\ \frac{\partial y}{\partial z_i} \end{array}\right] = \left[\begin{array}{c} \tan(a) \cos(k) \\ \tan(a) \sin(k) \end{array}\right].
\end{equation*}
It is now clear that
\begin{equation*}
\nu_i^i(q) \cdot n_i = \tan(a).
\end{equation*}

The evaluation of $n_i$, $\upsilon_i^i(q)$ and $\nu_i^i(q)$ on $\partial W_j \cap \partial W_i$ depends on the choice of $f(q)$ and the partitioning scheme described in Section \ref{section:partitioning}. 

When $f_i(q) = f_j(q)$, the evaluation of $n_i$, $\upsilon_i^i(q)$ and $\nu_i^i(q)$ is irrelevant since the corresponding integral will be $0$ due to the $f_i(q) - f_j(q)$ term.

When $f_i(q) > f_j(q)$, due to the partition scheme, the common sensed region between nodes $i$ and $j$ will be assigned to node $i$. As a result, $\partial W_i \cap \partial W_j$ will be an arc of $\partial C_i^s$ and will change when node $i$ moves. Thus the evaluation of $n_i$, $\upsilon_i^i(q)$ and $\nu_i^i(q)$ will be the same as their evaluation over $\partial W_i \cap \partial\mathcal{O}$ described previously.

When $f_i(q) < f_j(q)$, due to the partition scheme, the common sensed region between nodes $i$ and $j$ will be assigned to node $j$. As a result, $\partial W_i \cap \partial W_j$ will be an arc of $\partial C_j^s$ which will not change when node $i$ moves. Thus $n_i$, $\upsilon_i^i(q)$ and $\nu_i^i(q)$ are both zero in this case.

It is thus concluded that for the integrals over $\partial W_i \cap \partial W_j$ of node $i$, only arcs where $f_i(q) > f_j(q)$ need to be considered when integrating.

\bibliographystyle{IEEEtran}
%\bibliography{bibliography/ysphdbook,bibliography/refs1}
 %Generated by IEEEtran.bst, version: 1.14 (2015/08/26)

\end{document}